\documentclass[12pt, onecolumn, draftclsnofoot]{IEEEtran}
\usepackage[explicit]{titlesec}
\usepackage{graphicx}
\usepackage{verbatim}
\usepackage{epstopdf}
\usepackage{subfigure} 
\usepackage{algorithmic}
\usepackage{amsmath}
\usepackage{amsthm}
\usepackage{mathrsfs}
\usepackage{extarrows}
\usepackage{cite}
\newtheorem{theorem}{Theorem}

\usepackage{bbm}
\usepackage{amssymb}
\newtheorem{remark}{Remark}

\newcommand{\RNum}[1]{\uppercase\expandafter{\romannumeral #1\relax}}
\newtheorem{lemma}{Lemma}

\newtheorem{definition}{Definition}

\usepackage{stfloats}

\usepackage{color, soul}
\usepackage{tikz,xcolor,hyperref}

\definecolor{lime}{HTML}{A6CE39}

\titlespacing{\section}{0pt}{1.2ex plus .0ex minus .0ex}{.3ex plus .0ex}
\titlespacing{\subsection}{0pt}{1.2ex plus .0ex minus .0ex}{.3ex plus .0ex}

\makeatletter
\DeclareRobustCommand{\orcidicon}{%
	\begin{tikzpicture}
	\draw[lime, fill=lime] (0,0) 
	circle [radius=0.16] 
	node[white] {{\fontfamily{qag}\selectfont \tiny ID}};    \draw[white, fill=white] (-0.0625,0.095) 
	circle [radius=0.007];    \end{tikzpicture}
	\hspace{-2mm}}
\foreach \x in {A, ..., Z}{%
	\expandafter\xdef\csname orcid\x\endcsname{\noexpand\href{https://orcid.org/\csname orcidauthor\x\endcsname}{\noexpand\orcidicon}}
}

\newcommand*\bigcdot{\mathpalette\bigcdot@{.5}}
\newcommand*\bigcdot@[2]{\mathbin{\vcenter{\hbox{\scalebox{#2}{$\m@th#1\bullet$}}}}}
\makeatother

\usepackage[linesnumbered,ruled,vlined]{algorithm2e}

\begin{document}
\title{Goal-oriented Tensor: Beyond Age of Information Towards Semantics-Empowered Goal-Oriented Communications}
\author{Aimin Li\orcidB{}, 
	\emph{Graduate Student Member, IEEE,}
	 Shaohua Wu\orcidC{}, 
	 \emph{Member, IEEE,}\\
    Sumei Sun\orcidD{}, 
	 \emph{Fellow, IEEE}, and Jie Cao\orcidE{},
    \emph{Member, IEEE} }
%
%
%

\maketitle
\allowdisplaybreaks
\begin{abstract}
Optimizations premised on open-loop metrics such as Age of Information (AoI) indirectly enhance the system's decision-making \emph{utility}. We therefore propose a novel closed-loop metric named Goal-oriented Tensor (GoT) to directly quantify the impact of semantic mismatches on goal-oriented decision-making \emph{utility}. Leveraging the GoT, we consider a \emph{sampler \& decision-maker} pair that works collaboratively and distributively to achieve a shared goal of communications. We formulate a two-agent infinite-horizon Decentralized Partially Observable Markov Decision Process (Dec-POMDP) to conjointly deduce the optimal deterministic sampling policy and decision-making policy. To circumvent the \emph{curse of dimensionality} in obtaining an optimal deterministic joint policy through Brute-Force-Search, a sub-optimal yet computationally efficient algorithm is developed. This algorithm is predicated on the search for a Nash Equilibrium between the sampler and the decision-maker. Simulation results reveal that the proposed \emph{sampler \& decision-maker} co-design surpasses the current literature on AoI and its variants in terms of both goal achievement \emph{utility} and sparse sampling rate, signifying progress in the semantics-conscious, goal-driven sparse sampling design.  
\end{abstract}

\IEEEpeerreviewmaketitle

\section{Introduction}\label{sectionI}
The recent advancement of the emerging 5G and beyond has spawned the proliferation of both theoretical development and application instances for Internet of Things (IoT) networks. In such networks, timely status updates are becoming increasingly crucial for enabling real-time monitoring and actuation across a plethora of applications. To address the inadequacies of traditional throughput and delay metrics in such contexts, the \emph{Age of Information} (AoI) has emerged as an innovative metric to capture the data freshness perceived by the receiver \cite{kaul2011minimizing}, defined as $\mathrm{AoI}(t)=t-U(t)$,
where $U(t)$ denotes the generation time of the latest received packet before time $t$. Since its inception, AoI has garnered significant research attention and has been extensively analyzed in the queuing systems \cite{sun2017update,costa2016age,yates2018age,kam2015effect,bedewy2019age,moltafet2020age}, physical-layer communications \cite{Sac2018AgeOptimalCC, xie2020age,you2021age,meng2022analysis,pan2022age,Maatouk2019MinimizingTA,wu2022minimizing}, MAC-layer communications \cite{ceran2019average,9686609,Peng2020AgeOI, PanTVT2022,li2022age}, industrial IoT \cite{BZJSSU, XinXie2023}, energy harvesting systems\cite{AbdE2022,pappas,yates2015lazy,arafa2019age}, and etc. (see \cite{yatesAgeInformationIntroduction2021} and the references therein for more comprehensive review). These research efforts are driven by the consensus that a freshly received message typically contains more valuable information, thereby enhancing the \emph{utility} of decision-making processes.

Though AoI has been proven efficient in many freshness-critical applications, it exhibits several critical shortcomings. Specifically, ($a$) AoI does not provide a direct measure of information value; 
 ($b$) AoI does not consider the content dynamics of source data and ignores the effect of End-to-End (E2E) information mismatch on the decision-making process. 
 
 To address shortcoming (a), a typical approach is to impose a non-linear penalty on AoI\cite{sun2019sampling,kosta2020cost,cho2003effective,azar2018tractable,bastopcu2020information}. In \cite{sun2019sampling}, the authors map the AoI to a non-linear and non-decreasing function $f(\mathrm{AoI}(t))$ to evaluate the degree of ``\emph{discontent}'' resulting from \emph{stale} information. Subsequently, the optimal sampling policy is deduced for an arbitrary non-decreasing penalty function. The authors in \cite{kosta2020cost} introduce two penalty functions, namely the exponential penalty function $a^{\mathrm{AoI}(t)}-1$ and the logarithmic penalty function $\log_a{\left(\mathrm{AoI}(t+1)\right)}$, for evaluating the \emph{Cost of Update Delay} (CoUD). In \cite{cho2003effective} and \cite{azar2018tractable}, the binary indicator function $\mathbbm{1}_{\left\{\mathrm{AoI(t)>d}\right\}}$ is applied to evaluate whether the most recently received message is up-to-date. Specifically, the penalty assumes a value of $1$ when the instantaneous AoI surpasses a predetermined threshold $d$; otherwise, the penalty reverts to $0$. The freshness of web crawling is evaluated through this AoI-threshold binary indicator function. In \cite{bastopcu2020information}, an analogous binary indicator approach is implemented in caching systems to evaluate the freshness of information. 

The above works tend to tailor a particular non-linear penalty function to evaluate the information value. However, the selection of penalty functions in the above research relies exclusively on empirical configurations, devoid of rigorous derivations. To this end, several information-theoretic techniques strive to explicitly derive the non-linear penalty function in terms of AoI \cite{truong2013effects,8445873,wangFrameworkCharacterisingValue2021,9781393}. One such quintessential work is the auto-correlation function $\mathbb{E}\left[X_t^*X_{t-\mathrm{AoI}(t)}\right]$, which proves to be a non-linear function of AoI when the source is stationary \cite{truong2013effects}. Another methodology worth noting is the mutual information metric between the present source state $X_t$ and the vector consisting of an ensemble of successfully received updates $\mathbf{W}_t$ \cite{8445873,wangFrameworkCharacterisingValue2021}. In the context of a Markovian source, this metric can be reduced to $I(X_t;X_{t-\mathrm{AoI}(t)})$, which demonstrates a non-linear dependency on AoI under both the Gaussian Markov source and the Binary Markov source \cite{8445873}. In \cite{wangFrameworkCharacterisingValue2021}, an analogous approach is utilized to characterize the \emph{value of information} (VoI) for the Ornstein-Uhlenbeck (OU) process, which likewise demonstrates a non-linear dependency on AoI. In \cite{9781393}, the conditional entropy $H(X_t|\mathbf{W}_t)$ is further investigated to measure the uncertainty of the source for a remote estimator given the history received updates $\mathbf{W}_t$. When applied to a Markov Source, this conditional entropy simplifies to $H(X_t|{X}_{t-\mathrm{AoI}(t)})$, thus exemplifying a non-linear penalty associated with AoI.

To address shortcoming (b), substantial research efforts have been invested in the optimization of the \emph{Mean Squared Error} (MSE), denoted by $(X_t-\hat{X}_t)^2$, with an ultimate objective of constructing a real-time reconstruction remote estimation system \cite{kam2018towards,sun2019wiener,ornee2019sampling,arafa2021sample}. In \cite{kam2018towards}, a metric termed \emph{effective age} is proposed to minimize the MSE for the remote estimation of a Markov source. In \cite{sun2019wiener} and \cite{ornee2019sampling}, two Markov sources of interest, the Wiener process and the OU process are investigated to deduce the MSE-optimal sampling policy. Intriguingly, these policies were found to be threshold-based, activating sampling only when the instantaneous MSE exceeds a predefined threshold, otherwise maintaining a state of idleness. The authors in \cite{arafa2021sample} explored the trade-off between MSE and quantization over a noisy channel, and derived the MSE-optimal sampling strategy for the OU process.

Complementary to the above MSE-centered research, variants of AoI that address shortcoming $(b)$ have also been conceptualized \cite{AoCI,zheng2020urgency,zhong2018two,AoII}. In \cite{AoCI}, \emph{Age of Changed Information} (AoCI) is proposed to address the ignorance of content dynamics of AoI. In this regard, unchanged statuses do not necessarily provide new information and thus are not prioritized for transmission. In \cite{zheng2020urgency}, the authors introduce the context-aware weighting coefficient to propose the \emph{Urgency of Information} (UoI), a metric capable of measuring the weighted MSE in diverse urgency contexts. In \cite{zhong2018two}, the authors propose a novel age penalty named \emph{Age of Synchronization} (AoS), a novel metric quantifying the time since the most recent end-to-end synchronization. Moreover, considering that an E2E status mismatch may exert a detrimental effect on the overall system's performance over time, the authors of \cite{AoII} propose a novel metric called \emph{Age of Incorrect Information} (AoII). This metric quantifies the adverse effect stemming from the duration of the E2E mismatch, revealing that both the degree and duration of E2E semantic mismatches lead to a \emph{utility} reduction for subsequent decision-making. 

The above studies focused on the open-loop generation-to-reception process within a transmitter-receiver pair, neglecting the closed-loop perception-actuation timeliness. A notable development addressing this issue is the extension from Up/Down-Link (UL/DL) AoI to a closed-loop AoI metric, referred to as the \emph{Age of Loop} (AoL) \cite{10012674}. Unlike the traditional open-loop AoI, which diminishes upon successful reception of a unidirectional update, the AoL decreases solely when both the UL status and the DL command are successfully received. Another advanced metric in \cite{nikkhah2023age}, called Age of Actuation (AoA), also encapsulates the actuation timeliness, which proves pertinent when the receiver employs the received update to execute timely actuation.

\begin{figure*}
	\centering
\includegraphics[angle=0,width=0.9\textwidth]{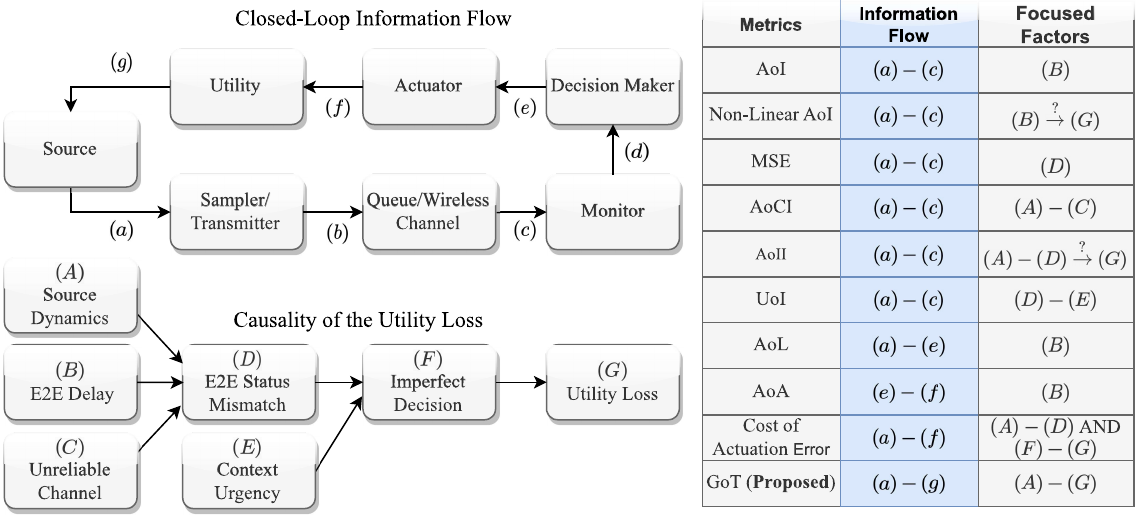}
	\caption{Interconnections of Age of Information and Beyond in the literature.}\label{Figure0}
\end{figure*}

Notwithstanding the above advancements, the question on \emph{how the E2E mismatch affects the closed-loop utility of decision-making has yet to be addressed}. To address this issue, \cite{Kountouris2020SemanticsEmpoweredCF,9551200,Salimnejad2023RealtimeRO,fountoulakisGoalorientedPoliciesCost2023} introduce a metric termed \emph{Cost of Actuation Error} to delve deeper into the cost resulting from the error actuation due to imprecise real-time estimations. Specifically, the \emph{Cost of Actuation Error} is denoted by an asymmetric zero diagonal matrix $\mathbf{C}$, with each value $C_{X_t,\hat{X}_t}$ representing the instant cost under the E2E mismatch status $(X_t,\hat{X}_t)_{X_t\ne\hat{X}_t}$. In this regard, the authors reveal that the \emph{utility} of decision-making bears a close relation to the E2E semantic mismatch category, as opposed to the mismatch duration (AoII) or mismatch degree (MSE). For example, an E2E semantic mismatch category that ``Fire'' is estimated as ``No Fire'' will result in high cost; while in the opposite scenario, the cost is low. Nonetheless, we notice that $i$) the method to obtain a \emph{Cost of Actuation Error} remains unclear, which implicitly necessitates a pre-established decision-making policy; $ii$) \emph{Cost of Actuation Error} does not consider the context-varying factors, which may also affect the decision-making \emph{utility}; $iii$) the zero diagonal property of the matrix implies the supposition that if $X_t=\hat{X}_t$, then $C_{X_t,\hat{X}_t}=0$, thereby signifying that error-less actuation necessitates no energy expenditure. Fig. \ref{Figure0} provides an overview of the existing metrics.

Against this background, the present authors have recently proposed a new metric referred to as GoT in \cite{li2023goaloriented}, which, compared to \emph{Cost of Actuation Error}, introduces new dimensions of the context $\Phi_{t}$ and the decision-making policy $\pi_A$ to describe the true \emph{utility} of decision-making. Remarkably, we find that GoT offers a versatile degeneration to established metrics such as AoI, MSE, UoI, AoII, and \emph{Cost of Actuation Error}. Additionally, it provides a balanced evaluation of the cost trade-off between the sampling and decision-making. The contributions of this work could be summarized as follows:

 $\bullet$ We focus on the decision \emph{utility} issue directly by employing the GoT. A controlled Markov source is observed, wherein the transition of the source is dependent on both the decision-making at the receiver and the contextual situation it is situated. In this case, the decision-making will lead to three aspects in \emph{utility}: $i$) the future evolution of the source; $ii$) the instant cost at the source; $iii$) the energy and resources consumed by actuation. 

 $\bullet$ We accomplish the goal-oriented \emph{sampler \& decision-maker} co-design, which, to the best of our knowledge, represents the first work that addresses the trade-off between semantics-aware sampling and goal-oriented decision-making. Specifically, we formulate this problem as a two-agent infinite-horizon Decentralized Partially Observable Markov Decision Process (Dec-POMDP) problem, with one agent embodying the semantics and context-aware sampler, and the other representing the goal-oriented decision-maker. Note that the optimal solution of even a finite-horizon Dec-POMDP is known to be NEXP-hard \cite{bernstein2002complexity}, we develop the RVI-Brute-Force-Search Algorithm. This algorithm seeks to derive optimal deterministic joint policies for both sampling and decision-making. A thorough discussion on the optimality of our algorithm is also presented.

 $\bullet$ To further mitigate the ``\emph{curse of dimensionality}'' intricately linked with the execution of the optimal RVI-Brute-Force-Search, we introduce a low-complexity yet efficient algorithm to solve the problem. The algorithm is designed by decoupling the problem into two distinct components: a Markov Decision Process (MDP) problem and a Partially Observable Markov Decision Process (POMDP) issue. Following this separation, the algorithm endeavors to search for the joint Nash Equilibrium between the sampler and the decision-maker, providing a sub-optimal solution to this goal-oriented communication-decision-making co-design.


\section{Goal-oriented Tensor}\label{section II}
\subsection{Specific Examples of Goal-Oriented Communications}
Consider a time-slotted communication system. Let $X_t\in\mathcal{S}$ represent the perceived semantic status at time slot $t$ at the source, and $\hat{X}_t\in\mathcal{S}$ denote the constructed estimated semantic status at time slot $t$. \emph{Real-time reconstruction-oriented} communications is a special type of goal-oriented communications, whose goal is achieving real-time accurate reconstruction:
\begin{equation}
\min\mathop {\lim \sup }\limits_{T \to \infty }\frac{1}{T}\sum_{t=0}^{T-1}(X_t-\hat{X}_t)^2.
\end{equation}

Although \emph{real-time reconstruction} is important, it does not represent the ultimate goals of communications, such as implementing accurate actuation (as opposed to merely \emph{real-time reconstruction}) and minimizing the system's long-term \emph{Cost of Actuation Error}, as in \cite{Kountouris2020SemanticsEmpoweredCF, 9551200, Salimnejad2023RealtimeRO, fountoulakisGoalorientedPoliciesCost2023}:

\begin{equation}
\min\mathop {\lim \sup }\limits_{T \to \infty }\frac{1}{T}\sum_{t=0}^{T-1}C_{X_t,\hat{X}_t},
\end{equation}
where $C_{X_t,\hat{X}_t}$ represents the instantaneous system cost of the system at time slot $t$. This cost is derived from the erroneous actuation stemming from the status mismatch between transceivers.

Following the avenue of \emph{Cost of Actuation Error}, we notice that the matrix-based metric could be further augmented to tensors to realize more flexible goal characterizations. For example, drawing parallels with the concept of \emph{Urgency of Information} \cite{zheng2020urgency}, we can introduce a context element $\Phi_t$ to incorporate context-aware attributes into this metric.\footnote{It is important to note that the GoT could be expanded into higher dimensions by integrating additional components, including actuation policies, task-specific attributes, and other factors.} Accordingly, we can define a three-dimensional GoT through a specified mapping:
$(X_t,\Phi_t,\hat{X}_t)\in \mathcal{S}\times \mathcal{V}\times\mathcal{S}\overset{\mathcal{L}}{\rightarrow} \mathrm{GoT}(t)\in\mathbb{R}$, which could be visualized by Fig. \ref{Figure1}. In this regard, the GoT, denoted by $\mathcal{L}(X_t,\Phi_t,\hat{X}_t)$ or $\mathrm{GoT}(t)$, indicates the instantaneous cost of the system at time slot $t$, with the knowledge of $(X_t,\Phi_t,\hat{X}_t)$. Consequently, the overarching goal of this system could be succinctly expressed as:
\begin{equation}
\min\mathop {\lim \sup }\limits_{T \to \infty }\frac{1}{T}\sum_{t=0}^{T-1}\mathrm{GoT}(t).
\end{equation}

\subsection{Degeneration to Existing Metrics}
In this subsection, we demonstrate, through visualized examples and mathematical formulations, that a three-dimension GoT can degenerate to existing metrics. Fig. \ref{Figure1} showcases a variety of instances of the GoT metric. 

\noindent
$\bullet$ \textbf{Degeneration to AoI}: AoI is generally defined as $\mathrm{AoI}(t) \triangleq t-\max\left\{G_i:D_i<t\right\}$, where $G_i$ is the generated time stamp of the $i$-th status update, $D_i$ represents the corresponding deliver time slot. Since AoI is known to be semantics-agnostic \cite{uysal2022semantic,Kountouris2020SemanticsEmpoweredCF}, the values in the tensor only depend on the freshness context $\Phi_t\triangleq\mathrm{AoI}(t)$. In this case, each tensor value given a determined context status $\Phi(t)$ is a constant, and the GoT is reduced to 
\begin{equation}
\mathrm{GoT}(t)=\mathcal{L}(X_t,\Phi_t,\hat{X}_t)\overset{(a)}{=}\Phi(t)=\mathrm{AoI}(t).
\end{equation}
where (a) indicates that AoI is semantics-agnostic. In this case, AoI refers to the context exactly. The process of reducing GoT to AoI is visually depicted in Fig. \ref{Figure1}(a).

\begin{figure*}
	\centering
	\includegraphics[angle=0,width=0.9\textwidth]{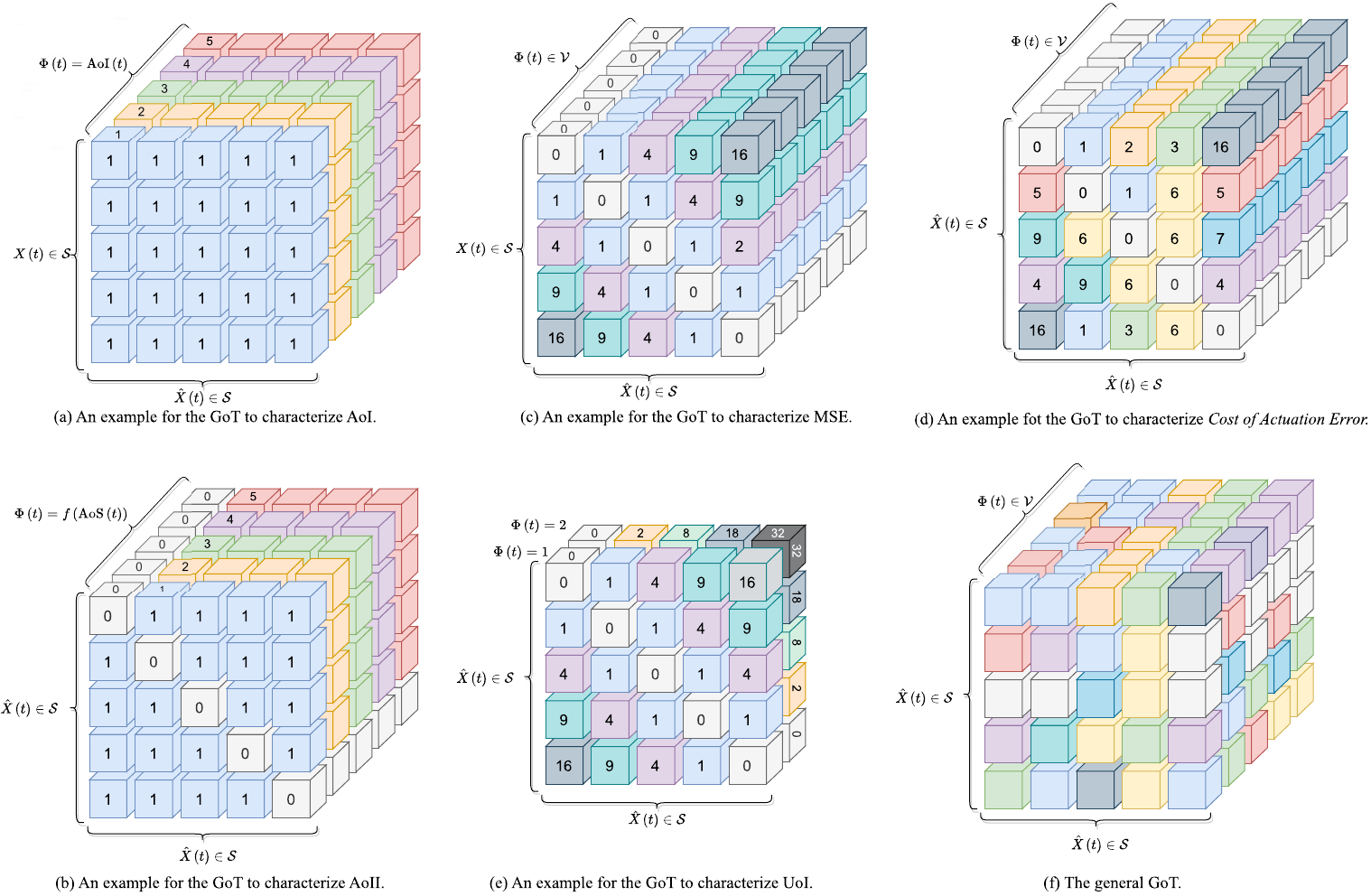}
	\caption{Visualize the GoT.}\label{Figure1}
\end{figure*}

\noindent
$\bullet$ \textbf{Degeneration to AoII}: AoII is defined as $\mathrm{AoII}(t) \triangleq f(\mathrm{AoS}(t))\cdot g(X_t,\hat{X}_t)$,
where $\textrm{AoS}(t) \triangleq t- \max\left\{\tau:t\le \tau, X_t=\hat{X}_t\right\}$. AoII embraces the semantics-aware characteristics and is hence regarded as an enabler of semantic communications \cite{9137714}. The inherent multiplicative characteristic of AoII guarantees the existence of a base layer of the tensor representation, from which other layers are derived by multiplying this base layer, as depicted in Fig. \ref{Figure1}(b). Let $\Phi_t\triangleq f(\mathrm{AoS}(t))$, the GoT is reduced to 
\begin{equation}
\mathrm{GoT}(t)=\mathcal{L}\left(X_t,\Phi_t,\hat{X}_t\right)\overset{(a)}{=}\Phi_t\cdot g\left(X_t,\hat{X}_t\right)=f(\mathrm{AoS}(t))\cdot g(X_t,\hat{X}_t)= \mathrm{AoII}(t),
\end{equation}
where $g(X_t,\hat{X}_t)$, typically characterized by $\mathbbm{1}_{\left\{X_t\ne\hat{X}_t\right\}}$, represents the base layer in the tensor, and $(a)$ indicates the inherent multiplicative characteristic of AoII. The visual representation of the GoT reduction to AoII is depicted in Fig. \ref{Figure1}(b).

\noindent
$\bullet$ \textbf{Degeneration to MSE}: MSE is defined as $\mathrm{MSE}(t) \triangleq  \left(X_t-\hat{X}_t\right)^2$. 
MSE is intuitively context-agnostic. In the scenario where $g(X_t,\hat{X}_t)=(X_t-\hat{X}_t)^2$, the GoT collapses to the MSE metric:
\begin{equation}
\mathrm{GoT}(t)=\mathcal{L}\left(X_t,\Phi_t,\hat{X}_t\right)\overset{(a)}{=}g\left(X_t,\hat{X}_t\right)
=\left(X_t-\hat{X}_t\right)^2=\mathrm{MSE}(t),
\end{equation}
where $(a)$ establishes due to the context-agnostic nature of MSE. The visualization of the GoT reduction to MSE is shown in Fig. \ref{Figure1}(c).

\noindent
$\bullet$ \textbf{Degeneration to UoI}: UoI is defined by $\mathrm{UoI}(t) \triangleq \Phi_t\cdot \left(X_t-\hat{X}_t\right)^2$, where the context-aware weighting coefficient $\Phi_t$ is further introduced \cite{zheng2020urgency}. When $g(X_t,\hat{X}_t)=(X_t-\hat{X}_t)^2$, the GoT could be transformed into the UoI by
\begin{equation}
\mathrm{GoT}(t)=\mathcal{L}\left(X_t,\Phi_t,\hat{X}_t\right)\overset{(a)}{=}\Phi_t\cdot g\left(X_t,\hat{X}_t\right)
=\Phi_t\cdot \left(X_t-\hat{X}_t\right)^2=\mathrm{UoI}(t),
\end{equation}
where $(a)$ indicates the inherent multiplicative characteristic of UoI. The visualization of the GoT reduction to UoI is shown in Fig. \ref{Figure1}(d).

\noindent
$\bullet$ \textbf{Degeneration to \emph{Cost of Actuation Error}}: \emph{Cost of Actuation Error} is defined by $C_{X_t,\hat{X}_t}$, which indicates the instantaneous system cost if the source status is $X_t$ and the estimated one $\hat{X}_t$ mismatch \cite{Salimnejad2023RealtimeRO}. Let $g(X_t,\hat{X}_t)=C_{X_t,\hat{X}_t}$, the GoT collapses to \emph{Cost of Actuation Error}:
\begin{equation}
\mathrm{GoT}(t)=\mathcal{L}\left(X_t,\Phi_t,\hat{X}_t\right)\overset{(a)}{=} g\left(X_t,\hat{X}_t\right)=C_{X_t,\hat{X}_t},
\end{equation}
where $(a)$ establishes due to the context-agnostic nature of \emph{Cost of Actuation Error}. The visualization of the GoT reduction to \emph{Cost of Actuation Error} is shown in Fig. \ref{Figure1}(e).
\subsection{Goal Characterization Through GoT}

As shown in Fig. \ref{Figure1}(f), a more general GoT exhibits neither symmetry, zero diagonal property, nor a base layer, offering considerable versatility contingent upon the specific goal. Here we propose a method that constructs a specific GoT, there are four steps: 

\noindent
$\bullet$ \textbf{Step 1:} Clarify the scenario and the communication goal. For instance, in the wireless accident monitoring and rescue systems, the goal is to minimize the cumulative average cost associated with accidents and their corresponding rescue interventions over the long term.

\noindent
$\bullet$ \textbf{Step 2:} Define the sets of semantic status, $\mathcal{S}$, and contextual status, $\mathcal{V}$. These sets can be modeled as collections of discrete components. For instance, the set $\mathcal{S}$ might encompass the gravity of industrial mishaps in intelligent factories, whereas the set $\mathcal{V}$ could encompass the meteorological circumstances, which potentially influence the source dynamics and the costs.

\noindent
$\bullet$ \textbf{Step 3:} The GoT could be decoupled by three factors: \footnote{The definitions of costs are diverse, covering aspects such as financial loss, resource usage, or a normalized metric derived through scaling, depending on the specific objective they are designed to address.}

 \noindent 
 $i)$ The status inherent cost $C_1(X_t,\Phi_t)$. It quantifies the cost associated with different status pairs $(X_t, \Phi_t)$ in the absence of external influences;

\noindent 
 $ii)$ The actuation gain $C_2(\pi_A(\hat{X}_t))$, where $\pi_A$ is a deterministic decision policy contingent upon $\hat{X}_t$. This cost quantifies the positive \emph{utility} resulting from the actuation $\pi_A(\hat{X}(t))$;

\noindent 
 $iii)$ The actuation resource expenditure $C_3(\pi_A(\hat{X}_t))$, which reflects the resources consumed by a particular actuation $\pi_A(\hat{X}(t))$.

\noindent
$\bullet$ \textbf{Step 4:} Constructing the GoT. The GoT, given a specific triple-tuple $(X_t, \hat{X}_t,\Phi_t)$ and a certain decision strategy $\pi_A$, is calculated by 
\begin{equation}\label{got}
\mathrm{GoT}^{\pi_A}(t)=\mathcal{L}(X_t,\Phi_t,\hat{X}_t,\pi_A)=\left[C_1(X_t,\Phi_t) - a C_2(\pi_A(\hat{X}_t))\right]^+
+ bC_3(\pi_A(\hat{X}_t)). 
\end{equation}
The ramp function $\left[\cdot\right]^+$ ensures that any actuation $\pi_A(\hat{X}_t)$ reduces the cost to a maximum of 0. A visualization of a specific GoT construction is shown in Fig. \ref{gotconstructing}. The GoT in Fig. \ref{gotconstructing} is obtained through the following definition:
\begin{equation}\label{gotexample}
\begin{array}{c}
C_1(X_t,\Phi_t)= \left(\begin{array}{c|ccc}
~ & 0 & 1 & 2 \\
\hline
0 & 0 & 1 & 3 \\
1 & 0 & 2 & 5
\end{array}\right),
\pi_A(\hat{X}_t)=\left(\begin{array}{ccc}
0 & 1 & 2\\
\hline
a_0 & a_1 & a_2\\
\end{array}\right), \\
C_2(a_A(t))=\left(\begin{array}{ccc}
a_0 & a_1 & a_2\\
\hline
0 & 2 & 4\\
\end{array}\right), 
C_3(a_A(t)) =\left(\begin{array}{ccc}
a_0 & a_1 & a_2\\
\hline
0 & 1 & 2\\
\end{array}\right).
\end{array}
\end{equation}

\section{System Model}\label{section III}
\begin{figure}
	\centering
\begin{minipage}[t]{0.48\textwidth}
		\centering
	\includegraphics[angle=0,width=1\textwidth]{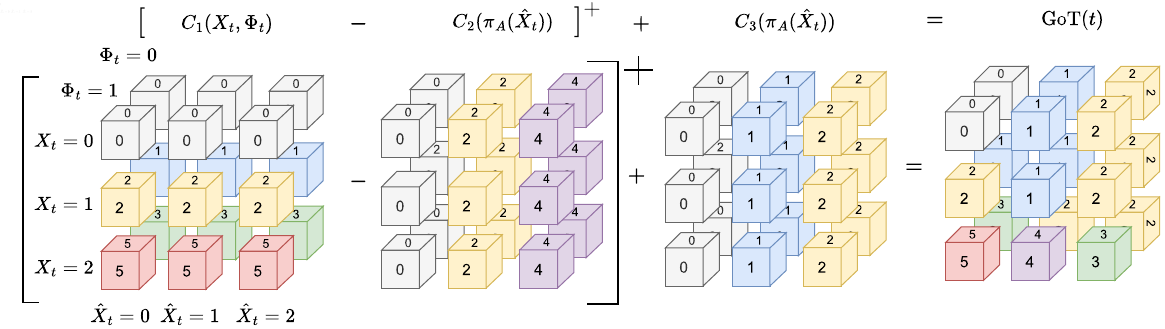}
	\caption{A visualized example for characterizing the GoT through (\ref{got}), where $[\cdot]^+$ represents the ramp function,}\label{gotconstructing}
\end{minipage}
\begin{minipage}[t]{0.48\textwidth}
		\centering	\includegraphics[angle=0,width=1\textwidth]{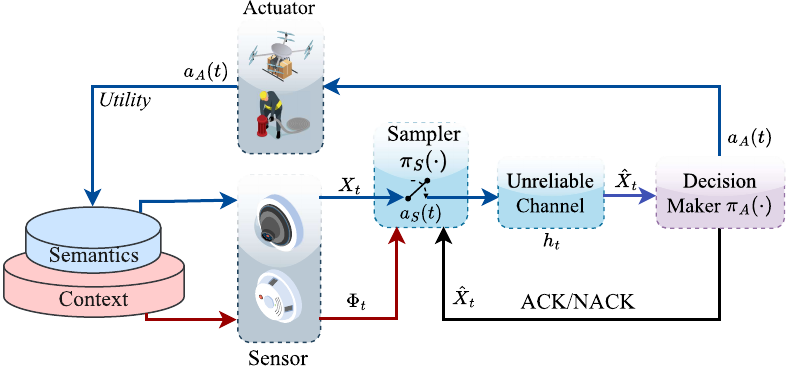}
	\caption{Illustration of our considered system where transmitted semantic status arrives at a receiver for decision-making to achieve a certain goal.}\label{systemmodel}
\end{minipage}
\end{figure}
We consider a time-slotted perception-actuation loop where both the perceived semantics $X_t\in\mathcal{S}=\left\{s_1,\cdots,s_{|\mathcal{S}|}\right\}$ and context $\Phi_t\in\mathcal{V}=\left\{v_1,\cdots,v_{|\mathcal{V}|}\right\}$ are input into a semantic sampler, tasked with determining the significance of the present status $X_t$ and subsequently deciding if it warrants transmission via an unreliable channel. The semantics and context are extracted and assumed to perfectly describe the status of the observed process. The binary indicator, $a_S(t)=\pi_S(X_t,\Phi_{t},\hat{X}_t) \in \left\{0, 1\right\}$, signifies the sampling and transmission action at time slot $t$, with the value $1$ representing the execution of sampling and transmission, and the value $0$ indicating the idleness of the sampler. $\pi_S$ here represents the sampling policy. We consider a perfect and delay-free feedback channel \cite{Kountouris2020SemanticsEmpoweredCF,9551200,Salimnejad2023RealtimeRO,fountoulakisGoalorientedPoliciesCost2023}, with ACK representing a successful transmission and NACK representing the otherwise. The decision-maker at the receiver will make decisions $a_A(t)\in\mathcal{A}_A=\left\{a_1,\cdots,a_{|\mathcal{A}_A|}\right\}$ based on the estimate $\hat{X}_t$\footnote{We consider a general abstract decision-making set $\mathcal{A}_A$ that exhibits adaptability to diverse applications. Notably, this decision-making set can be customized and tailored to suit specific requirements.}, which will ultimately affect the \emph{utility} of the system. An illustration of our considered model is shown in Fig. \ref{systemmodel}. 

\subsection{Semantics and Context Dynamics}

Thus far, a plethora of studies have delved into the analysis of various discrete Markov sources, encompassing Birth-Death Markov Processes elucidated in \cite{fountoulakisGoalorientedPoliciesCost2023}, binary Markov sources elucidated in \cite{kam2020age}, and etc. In real-world situations, the context and the actuation also affect the source's evolution. Consequently, we consider a context-dependent controlled Discrete Markov source:
\begin{equation}\label{Source}
\Pr \left( {X_{t+1}=s_u\left| {X_t=s_i, a_A(t)=a_m, \Phi_t=v_k} \right.} \right)=p_{i,u}^{(k,m)},
\end{equation}
where the dynamics of the source is dependent on both the decision-making $a_A(t)$ and context $\Phi_{t}$. Furthermore, we take into account the variations in context $\Phi_t$, characterized by:
\begin{equation}\label{context}
\Pr \left( {\Phi_{t+1}=v_r\left| {\Phi_{t}=v_k} \right.} \right)=p_{k,r}.
\end{equation}

 Note that the semantic status $X_t$ and context status $\Phi_t$ could be tailored according to the specific application scenario. In general, these two processes are independent with each other.
\subsection{Unreliable Channel and Estimate Transition}
\begin{figure}[t]
	\centering
	\includegraphics[angle=0,width=0.6\textwidth]{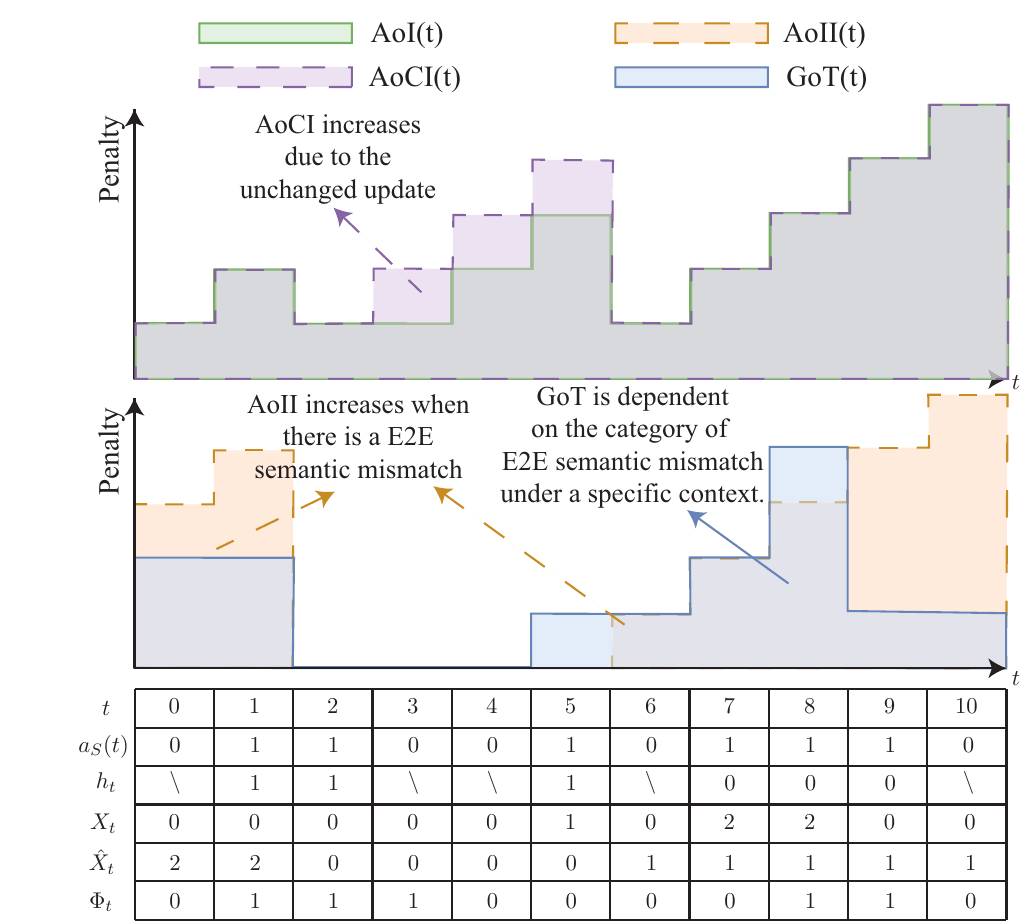}
	\caption{An illustration of AoI, AoCI, AoII, and GoT in a time-slotted status update system. Here, the value of GoT is obtained from the tensor obtained on the right-hand side of Fig. \ref{gotconstructing}.}\label{evolution}
\end{figure}
We assume that the channel realizations exhibit independence and identical distribution (i.i.d.) across time slots, following a Bernoulli distribution. Particularly, the channel realization $h_t$ assumes a value of $1$ in the event of successful transmission, and $0$ otherwise. Accordingly, we define the probability of successful transmission as $\Pr\left(h_t=1\right)=p_S$ and the failure probability as $\Pr\left(h_t=0\right)=1-p_S$. To characterize the dynamic process of $\hat{X}_t$, we consider two cases as described below:

\noindent $\bullet$ $a_S(t)=0$. In this case, the sampler and transmitter remain idle, manifesting that there is no new knowledge given to the receiver, \emph{i.e.}, $\hat{X}_{t+1}=\hat{X}_{t}$. As such, we have:
\begin{equation}\label{as0}
\Pr\left({\hat{X}_{t+1}=x\left|\hat{X}_t=s_j,a_S(t)=0\right.}\right)=\mathbbm{1}_{\left\{x=s_j\right\}}.
\end{equation}

\noindent $\bullet$ $a_S(t)=1$. In this case, the sampler and transmitter transmit the current semantic status $X_t$ through an unreliable channel. As the channel is unreliable, we differentiate between two distinct situations: $h_t=1$ and $h_t=0$:

\noindent (a) $h_t=1$. In this case, the transmission is successful. As such, the estimate at the receiver $\hat{X}_{t+1}$ is nothing but $X(t)$, and the transition probability is
\begin{equation}
\Pr\left({\hat{X}_{t+1}=x\left|\hat{X}_t=s_j,X_t=s_i,a_S(t)=1,h_t=1\right.}\right)=\mathbbm{1}_{\left\{x=s_i\right\}}.
\end{equation}

\noindent (b) $h_t=0$. In this case, the transmission is not successfully decoded by the receiver. As such, the estimate at the receiver $\hat{X}_{t+1}$ remains $\hat{X}(t)$. In this way, the transition probability is
\begin{equation}
\Pr\left({\hat{X}_{t+1}=x\left|\hat{X}_t=s_j,X_t=s_i,a_S(t)=1,h_t=0\right.}\right)=\mathbbm{1}_{\left\{x=s_j\right\}}.
\end{equation}
As the channel realization $h_t$ is independent with the process of $X_t$, $\hat{X}_t$, and $a_S(t)$, we have that
\begin{equation}\label{as1}
\begin{aligned}
&\Pr\left({\hat{X}_{t+1}=x\left|\hat{X}_t=s_j,X_t=s_i,a_S(t)=1\right.}\right)\\
&=\sum_{h_t}p\left(h_t\right)\Pr\left({\hat{X}_{t+1}=x\left|\hat{X}_t=s_j,X_t=s_i,a_S(t)=1,h_t\right.}\right)\\&=p_S\cdot\mathbbm{1}_{\left\{x=s_i\right\}}+(1-p_S)\cdot\mathbbm{1}_{\left\{x=s_j\right\}}.
\end{aligned}
\end{equation}
Combing (\ref{as0}) with (\ref{as1}) yields the dynamics of the estimate.
\subsection{Goal-oriented decision-making and Actuating}
We note that the previous works primarily focused on minimizing the open-loop freshness-related or error-related penalty for a transmitter-receiver system. Nevertheless, irrespective of the \emph{fresh} delivery or accurate end-to-end timely reconstruction, the ultimate goal of such optimization efforts is to ensure precise and effective decision-making. To this end, we broaden the open-loop transmitter-receiver information flow to include a perception-actuation closed-loop \emph{utility} flow by incorporating the decision-making and actuation processes. As a result, decision-making and actuation enable the conversion of status updates into ultimate effectiveness. Here the decision-making at time slot $t$ follows that $a_A(t)=\pi_A(\hat{X}_t)$, with $\pi_A$ representing the deterministic decision-making policy.

\section{Problem Formulation and Solution}\label{sectionIV}
Traditionally, the development of sampling strategies has been designed separately from the decision-making process. An archetypal illustration of this two-stage methodology involves first determining the optimal sampling policy based on AoI, MSE, and their derivatives, such as AoII, and subsequently accomplishing goal-oriented decision-making. This two-stage separate design arises from the inherent limitation of existing metrics that they fail to capture the closed-loop decision \emph{utility}. Nevertheless, the metric GoT empowers us to undertake a co-design of sampling and decision-making.

We adopt the \textit{team decision theory}, wherein two agents, one embodying the sampler and the other the decision-maker, collaborate to achieve a shared goal. We aim at determining a joint deterministic policy $\boldsymbol{\pi}_C=(\pi_S,\pi_A)$ that minimizes the long-term average cost of the system. It is considered that the sampling and transmission of an update also consumes energy, incurring a $C_S$ cost. In this case, the instant cost of the system could be clarified by $\mathrm{GoT}^{\pi_A}(t)+C_S\cdot a_S(t)$, and the problem is characterized as: 
\begin{equation}\label{p1}
\begin{array}{*{20}{c}}
{{{\cal P}} 1:}&{\mathop {\min }\limits_{{\boldsymbol{\pi}_C}  \in \Upsilon } \mathop {\lim \sup }\limits_{T \to \infty } \frac{1}{T}{\mathbb{E}^{{\boldsymbol{\pi}}_C} }\left( {\sum\limits_{t = 0}^{T - 1} {\mathrm{GoT}^{\pi_A}(t)+C_S\cdot a_S(t)} } \right)}\\
\end{array},
\end{equation}
where $\boldsymbol{\pi}_C=(\pi_S,\pi_A)$ denotes the joint sampling and decision-making policy, comprising $\pi_S=(a_S(0),a_S(1),\cdots)$ and $\pi_A=(a_A(0),a_A(1),\cdots)$, which correspond to the sampling action sequence and actuation sequence, respectively. Note that $\mathrm{GoT}^{\pi_A}(t)$ is characterized by (\ref{got}).
\subsection{Dec-POMDP Formulation}
The problem in (\ref{p1}) aims to find the optimal decentralized policy $\boldsymbol{\pi}_C$ so that the long-term average cost of the system is minimized. To solve the problem $\mathcal{P}1$, we ought to formulate a DEC-POMDP problem, which is initially introduced in \cite{bernstein2002complexity} to solve the cooperative sequential decision issues for distributed multi-agents. Within a Dec-POMDP framework, a team of agents cooperates to achieve a shared goal, relying solely on their localized knowledge. A typical Dec-POMDP is denoted by a tuple $\mathscr{M}_{DEC-POMDP}\triangleq\left\langle n, \mathcal{I}, \mathcal{A}, \mathcal{T}, \Omega, \mathcal{O}, \mathcal{R}  \right\rangle$:

\noindent $\bullet$ $n$ denotes the number of agents. In this instance, we have $n=2$, signifying the presence of two agents within this context: one agent $\mathcal{A}gent_S$ embodies the semantics-context-aware sampler and transmitter, while the other represents the $\hat{X}_t$-dependent decision-maker, denoted by $\mathcal{A}gent_A$.

\noindent $\bullet$ $\mathcal{I}$ is the finite set of the global system status, characterized by $(X_t, \hat{X}_t,\Phi_t)\in \mathcal{S}\times\mathcal{S}\times\mathcal{V}$. For the sake of brevity, we henceforth denote $\mathbf{W}_t=(X_t, \hat{X}_t,\Phi_t)$ in the squeal.

\noindent $\bullet$ $\mathcal{T}$ is the transition function defined by
\begin{equation}
\mathcal{T}\left(\mathbf{w},\mathbf{a},\mathbf{w}'\right)\triangleq\Pr(\mathbf{W}_{t+1}=\mathbf{w}'|\mathbf{W}_t=\mathbf{w},\mathbf{a}_t=\mathbf{a}),
\end{equation}
which is defined by the transition probability from global status $\mathbf{W}_t=\mathbf{w}$ to status $\mathbf{W}_{t+1}=\mathbf{w}'$, after the agents in the system taking a joint action $\mathbf{a}_t=\mathbf{a}=(a_S(t),a_A(t))$. {For the sake of concise notation, we let $p(\mathbf{w}'|\mathbf{w},\mathbf{a})$ symbolize $\mathcal{T}\left(\mathbf{w},\mathbf{a},\mathbf{w}'\right)$ in the subsequent discourse.} Then, by taking into account the \textit{conditional independence} among $X_{t+1}$, $\Phi_{t+1}$, and $\hat{X}_{t+1}$, given $(X_{t},\Phi_{t},\hat{X}_{t})$ and $\mathbf{a}(t)$, the transition functions can be calculated in lemma \ref{as11}.
\begin{lemma}\label{as11} The transition functions of the Dec-POMDP:
	\begin{equation}
	p\left((s_u,x,v_r)\left|(s_i,s_j,v_k),(1,a_m)\right.\right)=
	p_{i,u}^{(k,m)}\cdot p_{k,r} \cdot \left(p_S\cdot\mathbbm{1}_{\left\{x=s_i\right\}}+(1-p_S)\cdot\mathbbm{1}_{\left\{x=s_j\right\}}\right),
	\end{equation}
	\begin{equation}\label{as00}
	p\left((s_u,x,v_r)\left|(s_i,s_j,v_k),(0,a_m)\right.\right)=
	p_{i,u}^{(k,m)}\cdot p_{k,r} \cdot \mathbbm{1}_{\left\{x=s_j\right\}},
	\end{equation}
	for any $x\in\mathcal{S}$ and indexes $i$, $j$, $u\in\left\{1,2,\cdots,|\mathcal{S}|\right\}$, $k$, $r\in\left\{1,2,\cdots,|\mathcal{V}|\right\}$, and $m\in\left\{1,2,\cdots,|\mathcal{A}_A|\right\}$.
\end{lemma}
\begin{proof}
	The transition function can be derived by incorporating the dynamics in equations (\ref{Source}), (\ref{context}), (\ref{as0}), and (\ref{as1}). A more comprehensive proof is presented in Appendix \ref{prooflemma1}.
\end{proof}

\noindent $\bullet$ $\mathcal{A}=\mathcal{A}_S\times\mathcal{A}_A$, with $\mathcal{A}_S\triangleq\left\{0,1\right\}$ representing the set of binary sampling actions executed by the sampler, and $\mathcal{A}_A\triangleq\left\{a_0,\cdots,a_{M-1}\right\}$ representing the set of decision actions undertaken by the actuator.

\noindent $\bullet$ $\Omega=\Omega_S\times \Omega_A$ constitutes a finite set of joint observations. Generally, the observation made by a single agent regarding the system status is partially observable. $\Omega_S$ signifies the sampler's observation domain. In this instance, the sampler $\mathcal{A}gent_S$ is entirely observable, with $\Omega_S$ encompassing the comprehensive system state ${o}_S^{(t)}=\mathbf{W}_t$. $\Omega_A$ signifies the actuator's observation domain. In this case, the actuator (or decision-maker) $\mathcal{A}gent_A$ is partially observable, with $\Omega_A$ comprising $o_A^{(t)}=\hat{X}(t)$. The joint observation at time instant $t$ is denoted by $\mathbf{o}_t=(o_S^{(t)},o_A^{(t)})$.

\noindent $\bullet$ $\mathcal{O}=\mathcal{O}_S\times\mathcal{O}_A$ represents the observation function, where $\mathcal{O}_S$ and $\mathcal{O}_A$ denotes the observation function of the sampler $\mathcal{A}gent_S$ and the actuator $\mathcal{A}gent_A$, respectively, defined as:
\begin{equation}
\begin{aligned}
\mathcal{O}(\mathbf{w}, \mathbf{o})\triangleq\Pr(\mathbf{o}_t=\mathbf{o}|\mathbf{W}_t=\mathbf{w}),\\
\mathcal{O}_S(\mathbf{w}, o_S)\triangleq\Pr(o_S^{(t)}=o_S|\mathbf{W}_t=\mathbf{w}),\\
\mathcal{O}_A(\mathbf{w}, o_A)\triangleq\Pr(o_A^{(t)}=o_A|\mathbf{W}_t=\mathbf{w}).
\end{aligned}
\end{equation}	
The observation function of an agent $\mathcal{A}gent_i$ signifies the conditional probability of agent $\mathcal{A}gent_i$ perceiving $o_i$, contingent upon the prevailing global system state as $\mathbf{W}_t=\mathbf{w}$. For the sake of brevity, we henceforth let $p_A(o_A|\mathbf{w})$ represent $\mathcal{O}_A(\mathbf{w}, o_A)$ and $p_S(o_S|\mathbf{w})$ represent $\mathcal{O}_S(\mathbf{w}, o_A)$ in the subsequent discourse. In our considered model, the observation functions are deterministic, characterized by lemma \ref{obf}.
\begin{lemma} The observation functions of the Dec-POMDP: \label{obf}
	\begin{equation} 
	\begin{aligned}
	p_S\left((s_u,s_r,v_q)|(s_i,s_j,v_k)\right)&=\mathbbm{1}_{\left\{(s_u,s_r,v_q)=(s_i,s_j,v_k)\right\}}\\
	p_A\left(s_z|(s_i,s_j,v_k)\right)&=\mathbbm{1}_{\left\{s_z=s_j\right\}}\cdot
	\end{aligned}
	\end{equation}
	for indexes $z$, $i$, $j$, $u$, $r\in\left\{1,2,\cdots\,|\mathcal{S}|\right\}$, and $k$, $q\in\left\{1,2,\cdots\,|\mathcal{V}|\right\}$.
\end{lemma}

\noindent $\bullet$ $\mathcal{R}$ is the reward function, characterized as a mapping $\mathcal{I}\times\mathcal{A}\rightarrow\mathbb{R}$. In the long-term average reward maximizing setup, resolving a Dec-POMDP is equivalent to addressing the following problem ${\mathop {\min }\limits_{{\boldsymbol{\pi}_C}  \in \Upsilon } \mathop {\lim \sup }\limits_{T \to \infty } \frac{1}{T}{\mathbb{E}^{{\boldsymbol{\pi}}_C} }\left( -{\sum\limits_{t = 0}^{T - 1} r(t) } \right)}$.
Subsequently, to establish congruence with the problem in (\ref{p1}), the reward function is defined as:
\begin{equation}
r(t)=\mathcal{R}^{\pi_A}(\mathbf{w},a_S)=-\mathrm{GoT}^{\pi_A}(t)-C_S\cdot a_S(t).
\end{equation}

\subsection{Solutions to the Infinite-Horizon Dec-POMDP}

In general, solving a Dec-POMDP is known to be NEXP-complete for the finite-horizon setup \cite{bernstein2002complexity}, signifying that it necessitates formulating a conjecture about the solution non-deterministically, while each validation of a conjecture demands exponential time. For an infinite-horizon Dec-POMDP problem, finding an optimal policy for a Dec-POMDP problem is known to be undecidable. Nevertheless, within our considered model, both the sampling and decision-making processes are deterministic, given as $a_S(t)=\pi_S(\mathbf{w})$ and $a_A(t)=\pi_A(o_A)$. In such a circumstance, it is feasible to determine a joint optimal deterministic policy via Brute-Search-across the decision-making policy space.
\subsubsection{Optimal Solution}

The idea is based on the finding that, given a deterministic decision-making policy $\pi_A$, the sampling problem can be formulated as a standard fully observed MDP problem denoted by $\mathscr{M}^{\pi_A}_{\mathrm{MDP}}\triangleq\langle\mathcal{I},\mathcal{T}^{\pi_A},\mathcal{A}_S,\mathcal{R}\rangle$. 

\begin{definition}\label{proposition1}
	Given a deterministic decision-making policy $\pi_A$, the optimal sampling problem could be formulated by a typical fully observed MDP problem $\mathscr{M}^{\pi_A}_{\mathrm{MDP}}\triangleq\langle\mathcal{I},\mathcal{A}_S,\mathcal{T}_{\mathrm{MDP}}^{\pi_A},\mathcal{R}\rangle$, where the elements are given as follows:
	\begin{itemize}
	\item $\mathcal{I}$: the same as the pre-defined Dec-POMDP tuple.
	\item $\mathcal{A}_S=\left\{0,1\right\}$: the sampling and transmission action set.
	\item $\mathcal{T}^{\pi_A}$: the transition function given a deterministic decision-making policy $\pi_A$, which is 
	\begin{equation}
	\begin{aligned}
	\mathcal{T}^{\pi_A}(\mathbf{w},a_S,\mathbf{w}')=p^{\pi_A}\left(\mathbf{w}'|\mathbf{w},a_S\right)=\sum_{o_A\in\mathcal{O}_A}p\left(\mathbf{w}'|\mathbf{w},(a_S,\pi_A(o_A))\right)p_A(o_A|\mathbf{w})
	\end{aligned},
	\end{equation}
	where $p\left(\mathbf{w}'|\mathbf{w},(a_S,\pi_A(o_A))\right)$ could be obtained by Lemma \ref{as11} and $p(o_A|\mathbf{w})$ could be obtained by Lemma \ref{obf}.
	\item $\mathcal{R}$: the same as the pre-defined Dec-POMDP tuple.
	\end{itemize}
\end{definition}

We now proceed to solve the MDP problem $\mathscr{M}^{\pi_A}_{\mathrm{MDP}}$, which is characterized by a tuple $\langle\mathcal{I},\mathcal{T}^{\pi_A},\mathcal{A}_S,\mathcal{R}\rangle$. In order to deduce the optimal sampling policy under a deterministic decision-making policy $\pi_A$, it is imperative to resolve the Bellman equations \cite{bertsekas2012dynamic}:
\begin{equation}
\begin{aligned}
&\theta^*_{\pi_A}+V_{\pi_A}(\mathbf{w})=\mathop{\max}\limits_{a_S\in\mathcal{A}_A}\left\{\mathcal{R}^{\pi_A}(\mathbf{w},a_S)+\sum_{\mathbf{w}'\in\mathcal{I}}p(\mathbf{w}'|\mathbf{w},a_S)V_{\pi_A}(\mathbf{w}')\right\},\\
\end{aligned}
\end{equation}
where $V^{\pi_A}(\mathbf{w})$ is the value function and $\theta^*_{\pi_A}$ is the optimal long-term average reward given the decision-making policy $\pi_A$. We apply the relative value iteration (RVI) algorithm to solve this problem. The details are shown in Algorithm \ref{Algorithm 1}:
\begin{algorithm}
	\small
	\caption{The RVI Algorithm to Solve the MDP Given $\pi_A$}
	\label{Algorithm 1}
	\LinesNumbered
	\KwIn{The MDP tuple $\langle\mathcal{I},\mathcal{A}_S,\mathcal{T}^{\pi_A},\mathcal{R}\rangle$, $\epsilon$, $\pi_A$;}
	Initialization: $\forall\mathbf{w}\in\mathcal{I}$, $\tilde{V}^0_{\pi_A}(\mathbf{w})=0$, $\tilde{V}^{-1}_{\pi_A}(\mathbf{w})=\infty$, $k=0$ \;
	Choose $\mathbf{w}^{ref}$ arbitrarily\;
		\While {$||\tilde{V}_{\pi_A}^k(\mathbf{w})-\tilde{V}_{\pi_A}^{k-1}(\mathbf{w})||\ge \epsilon$}
		{
			$k=k+1$\;
		\For{$\mathbf{w}\in\mathcal{I}-\mathbf{w}^{ref}$}
		{
			$\begin{scriptsize}
			\begin{aligned}
			&\tilde{V}_{\pi_A}^k(\mathbf{w})=-g_k+\\
			&\mathop{\max}\limits_{a_S}\left\{\mathcal{R}(\mathbf{w},a_S)+\sum_{\mathbf{w}'\in\mathcal{I}-\mathbf{w}^{ref}}p(\mathbf{w}'|\mathbf{w},a_S)\tilde{V}^{k-1}_{\pi_A}(\mathbf{w}')\right\};
			\end{aligned}
			\end{scriptsize}$
		}			
		}
	$\begin{aligned}
		&\theta^*(\pi_A,\pi_S^*)=-\tilde{V}^{k}_{\pi_A}(\mathbf{w})\\
		&\mathop{\max}\limits_{a_S\in\mathcal{A}_S}\left\{\mathcal{R}(\mathbf{w},a_S)+\sum_{\mathbf{w}'\in\mathcal{I}}p(\mathbf{w}'|\mathbf{w},a_S)\tilde{V}^{k}_{\pi_A}(\mathbf{w}')\right\}
	\end{aligned}$\;
	\For{$\mathbf{w}\in\mathcal{I}$}{
	$\pi_S^{*}(\pi_A,\mathbf{w})=\mathop{\arg\max}\limits_{a_S}\left\{\mathcal{R}(\mathbf{w},a_S)+\sum_{\mathbf{w}'\in\mathcal{I}}p(\mathbf{w}'|\mathbf{w},a_S)\tilde{V}^{k}_{\pi_A}(\mathbf{w}')\right\}$;}
	\KwOut{$\pi_S^{*}(\pi_A)$, $\theta^*(\pi_A,\pi_S^*)$}
\end{algorithm}

With Definition \ref{proposition1} and Algorithm \ref{Algorithm 1} in hand, we could then perform a Brute-Force-Search across the decision-making policy space $\Upsilon_A$, thereby acquiring the joint sampling-decision-making policy. The algorithm is called RVI-Brute-Force-Search Algorithm, which is elaborated in Algorithm 2. In the following theorem, we discuss the optimality of the RVI-Brute-Force-Search Algorithm.
\begin{theorem}
	The RVI-Brute-Force-Search Algorithm (Algorithm \ref{Algorithm 2}) could achieve the optimal joint deterministic policies $(\pi_S^*,\pi_A^*)$, given that the transition function $\mathcal{T}^{\pi_A}$ follows a unichan. 
\end{theorem}
	\noindent \emph{Proof.} If the the transition function $\mathcal{T}^{\pi_A}$ follows a unichan, we obtain from \cite[Theorem 8.4.5]{puterman2014markov} that for any $\pi_A$, we could obtain the optimal deterministic policy $\pi_S^*$ such that
	$\theta^*({\pi_A},\pi_S^*)\le\theta^*({\pi_A},\pi_S)$. Also, Algorithm 2 assures that for any $\pi_A$, $\theta^*({\pi^*_A},\pi_S^*)\le\theta^*({\pi_A},\pi_S^*)$. This leads to the conclusion that for any $\boldsymbol{\pi}_C=(\pi_S,\pi_A)\in \Upsilon$, we have that
	\begin{equation}
	\theta^*({\pi^*_A},\pi_S^*)\le\theta^*({\pi_A},\pi_S^*)\le\theta^*({\pi_A},\pi_S).
	\end{equation}
\vspace{-1em}

Nonetheless, the Brute-Force-Search across the decision-making policy space remains computationally expensive, as the size of the decision-making policy space $\Upsilon_A$ amounts to $|\Upsilon_A|=\mathcal{A}_A^{\mathcal{O}_A}$. This implies that executing the RVI algorithm $\mathcal{A}_A^{\mathcal{O}_A}$ times is necessary to attain the optimal policy. Consequently, although proven to be optimal, such an algorithm is ill-suited for scenarios where $\mathcal{O}_A$ and $\mathcal{A}_A$ are considerably large. To ameliorate this challenge, we propose a sub-optimal, yet computation-efficient alternative in the subsequent section.
\subsubsection{A Sub-optimal Solution}
The method here is to instead find a locally optimal algorithm to circumvent the high complexity of the Brute-Force-Search-based approach. We apply the Joint Equilibrium-Based Search for Policies (JESP) for \emph{Nash equilibrium} solutions \cite{2003Taming}. Within this framework, the sampling policy is optimally responsive to the decision-making policy and vice versa, \emph{i.e.}, $\forall \pi_S, \pi_A,\theta(\pi_S^*,\pi_A^*)\le\theta(\pi_S,\pi_A^*),\theta(\pi_S^*,\pi_A^*)\le\theta(\pi_S^*,\pi_A)$. 

\begin{algorithm}
	\small
	\caption{The RVI-Brute-Force-Search Algorithm}
	\label{Algorithm 2}
	\LinesNumbered
	\KwIn{The Dec-POMDP tuple $\mathscr{M}_{DEC-POMDP}\triangleq\left\langle n, \mathcal{I}, \mathcal{A}, \mathcal{T}, \Omega, \mathcal{O}, \mathcal{R}  \right\rangle$;}
	\For{$\pi_A\in\Upsilon$}{
		Formulate the MDP problem $\mathscr{M}^{\pi_A}_{\mathrm{MDP}}\triangleq\langle\mathcal{I},\mathcal{A}_S,\mathcal{T}_{\mathrm{MDP}}^{\pi_A},\mathcal{R}\rangle$ as given in Definition \ref{proposition1}\;
		Run Algorithm 1 to obtain $\pi_S^*(\pi_A)$ and $\theta^*({\pi_A},\pi_S^*)$\;}
	Calculate the optimal joint policy:
	$\begin{cases}
	&\pi_A^* =\mathop{\arg\min}\nolimits_{\pi_A}\theta_{\pi_A}^* \\
	&\pi_S^*=\pi_S(\pi^*_A)
	\end{cases}$\;
	\KwOut{$\pi_S^{*}$, $\pi_A^*$}
\end{algorithm}

To search for the \emph{Nash equilibrium}, we first search for the optimal sampling policy prescribed a decision-making policy. This problem can be formulated as a standard fully observed MDP problem denoted by $\mathscr{M}^{\pi_A}_{\mathrm{MDP}}\triangleq\langle\mathcal{I},\mathcal{A}_S,\mathcal{T}_{\mathrm{MDP}}^{\pi_A},\mathcal{R}\rangle$ (see Definition \ref{proposition1}). Next, we alternatively fix the sampling policy $\pi_S$ and solve for the optimal decision-making policy $\pi_A$. This problem can be modeled as a memoryless partially observable Markov decision process (POMDP), denoted by $\mathscr{M}^{\pi_S}_{\mathrm{POMDP}}\triangleq\langle\mathcal{I},\mathcal{A}_A,\Omega_A,\mathcal{O}_A,\mathcal{T}^{\pi_S}_{\mathrm{POMDP}},\mathcal{R}\rangle$ (see Definition \ref{proposition2}). Then, by alternatively iterating between $\mathcal{A}gent_S$ and $\mathcal{A}gent_A$, we could obtain the \emph{Nash equilibrium} between the two agents.

\begin{definition}\label{proposition2}
	Given a deterministic sampling policy $\pi_S$, the optimal sampling problem could be formulated as a memoryless POMDP problem $\mathscr{M}^{\pi_S}_{\mathrm{POMDP}}\triangleq\langle\mathcal{I},\mathcal{A}_A,\Omega_A,\mathcal{O}_A,\mathcal{T}^{\pi_S}_{\mathrm{POMDP}},\mathcal{R}\rangle$, where the elements are given as follows:
	\begin{itemize}
		\item $\mathcal{I}$, $\Omega_A$, $\mathcal{A}_A$, and $\mathcal{O}_A$: the same as the pre-defined Dec-POMDP tuple.
		\item $\mathcal{T}^{\pi_S}_{\mathrm{POMDP}}$: the transition function given a deterministic sampling policy $\pi_S$, which is 
		\begin{equation}\label{TransitionPOMDP}
		\begin{aligned}
		\mathcal{T}^{\pi_S}_{\mathrm{POMDP}}(\mathbf{w},a_A,\mathbf{w}')=p^{\pi_S}\left(\mathbf{w}'|\mathbf{w},a_A\right)=p\left(\mathbf{w}'|\mathbf{w},(\pi_S(\mathbf{w}),a_A)\right)
		\end{aligned},
		\end{equation}
		where $p\left(\mathbf{w}'|\mathbf{w},(\pi_S(\mathbf{w}),a_A)\right)$ could be obtained by Lemma \ref{as11}.
		\item $\mathcal{R}$: the reward function is denoted as $\mathcal{R}^{\pi_S}(\mathbf{w},a_A)$, which could be obtained by (\ref{got}).
	\end{itemize}
\end{definition}
We then proceed to solve the memoryless POMDP problem discussed in Definition \ref{proposition2} to obtain the deterministic decision-making policy. Denote $p_{\pi_A}^{\pi_S}(\mathbf{w}'|\mathbf{w})$ as the transition probability $\Pr\left\{\mathbf{W}_{t+1}=\mathbf{w}'|\mathbf{W}_t=\mathbf{w}\right\}$ under the sampling policy $\pi_A$ and $\pi_S$, we then have that
\begin{equation}\label{transitionp2}
	p_{\pi_A}^{\pi_S}(\mathbf{w}'|\mathbf{w})=\sum_{o_A\in\mathcal{O}_A}p(o_A|\mathbf{w})\sum_{a_A\in\mathcal{A}_A}p^{\pi_S}(\mathbf{w}'|\mathbf{w},a_A)\pi_A(a_A|o_A),
\end{equation}
where $p(o_A|\mathbf{w})$ could be obtained by Lemma \ref{as11} and $p^{\pi_S}(\mathbf{w}'|\mathbf{w},\pi_A(o_A))$ is obtained by (\ref{TransitionPOMDP}). By assuming the ergodicity of the $p^{\pi_S}_{\pi_A}(\mathbf{w}'|\mathbf{w})$ and rewrite it as a matrix $\mathbf{P}_{\pi_A}^{\pi_S}$, we could then solve out the stationary distribution of the system status under the policies $\pi_S$ and $\pi_A$, denoted as $\boldsymbol{\mu}_{\pi_A}^{\pi_S}$, by solving the balance equations:
\begin{equation}\label{station}
\boldsymbol{\mu}_{\pi_A}^{\pi_S} \mathbf{P}_{\pi_A}^{\pi_S}=\boldsymbol{\mu}_{\pi_A}^{\pi_S} ,{~}\boldsymbol{\mu}_{\pi_A}^{\pi_S}\mathbf{e}=1,
\end{equation}
where $\mathbf{e}$ is the all one vector $[1,\cdots,1]_{|\mathcal{I}|\times1}$, $\boldsymbol{\mu}_{\pi_A}^{\pi_S}$ could be solved out by {Cramer's rules}. Denote ${\mu}_{\pi_A}^{\pi_S}(\mathbf{w})$ as the stationary distribution of $\mathbf{w}$. Also, we denote $r_{\pi_A}^{\pi_S}(\mathbf{w})$ as the expectation reward of global system status $\mathbf{w}$ under policies $\pi_A$ and $\pi_S$. It could be calculated as:
\begin{equation}\label{stationreward}
r_{\pi_A}^{\pi_S}(\mathbf{w})=\sum_{o_A\in\mathcal{O}_A}p(o_A|\mathbf{w})\mathcal{R}^{\pi_S}(\mathbf{w},\pi_A(o_A)).
\end{equation}
The performance measure is the long-term average reward:
\begin{equation}\label{averagereward}
\eta_{\pi_A}^{\pi_S}=
{\mathop {\lim \sup }\limits_{T \to \infty } \frac{1}{T}{\mathbb{E}^{(\pi_S,\pi_A)} }\left({\sum\limits_{t = 0}^{T - 1} r(t) } \right)}=\sum_{\mathbf{w}\in\mathcal{I}}{\mu}_{\pi_A}^{\pi_S}(\mathbf{w})\cdot r_{\pi_A}^{\pi_S}(\mathbf{w})
\end{equation}
With $\eta_{\pi_A}^{\pi_S}$ in hand, we then introduce the relative reward $g_{\pi_A}^{\pi_S}(\mathbf{w})$, defined by
\begin{equation}
g_{\pi_A}^{\pi_S}(\mathbf{w})\triangleq{\mathop {\lim \sup }\limits_{T \to \infty } \frac{1}{T}{\mathbb{E}^{(\pi_S,\pi_A)} }\left[{\sum\limits_{t = 0}^{T - 1} \left(r(t)-\eta_{\pi_A}^{\pi_S}\right) }\left|\mathbf{W}_0=\mathbf{w}\right. \right]},
\end{equation}
which satisfies the Poisson equations \cite{bertsekas1995neuro}:
\begin{equation}\label{poisson}
\eta_{\pi_A}^{\pi_S}+g_{\pi_A}^{\pi_S}(\mathbf{w})=r_{\pi_A}^{\pi_S}(\mathbf{w})+\sum_{\mathbf{w}'\in\mathcal{I}}p_{\pi_A}^{\pi_S}(\mathbf{w}'|\mathbf{w})g_{\pi_A}^{\pi_S}\mathbf{(w}').
\end{equation}
Denote $\mathbf{g}^{\pi_S}_{\pi_A}$ as the vector consisting of $g_{\pi_A}^{\pi_S}(\mathbf{w})$, $\mathbf{r}^{\pi_S}_{\pi_A}$ as the vector consisting of ${r}^{\pi_S}_{\pi_A}(\mathbf{w}),\mathbf{w}\in\mathcal{I}$. $\mathbf{g}^{\pi_S}_{\pi_A}$ could be solved by utilizing \cite{puterman2014markov}:
\begin{equation}\label{iterations}
\mathbf{g}^{\pi_S}_{\pi_A}=\left[(I-\mathbf{P}^{\pi_S}_{\pi_A}+\mathbf{e}\boldsymbol{\mu}^{\pi_S}_{\pi_A})^{-1}-\mathbf{e}\boldsymbol{\mu}^{\pi_S}_{\pi_A}\right]\mathbf{r}^{\pi_S}_{\pi_A}
\end{equation}

With the relative reward $g_{\pi_A}^{\pi_S}$ in hand, we then introduce $Q_{\pi_A}^{\pi_S}(\mathbf{w},a_A)$ and $Q_{\pi_A}^{\pi_S}(o_A,a_A)$ as follows:
\begin{lemma}$Q_{\pi_A}^{\pi_S}(\mathbf{w},a_A)$ and $Q_{\pi_A}^{\pi_S}(o_A,a_A)$ are defined and calculated as:
	\begin{equation}
	\begin{aligned}
	&Q_{\pi_A}^{\pi_S}(\mathbf{w},a_A)\triangleq{\mathop {\lim \sup }\limits_{T \to \infty } \frac{1}{T}{\mathbb{E}^{(\pi_S,\pi_A)} }\left[{\sum\limits_{t = 0}^{T - 1} \left(r(t)-\eta_{\pi_A}^{\pi_S}\right) }\left|\mathbf{W}_0=\mathbf{w},a_A(0)=a_A\right. \right]}\\
	&=\mathcal{R}^{\pi_S}(\mathbf{w},a_A)-\eta_{\pi_A}^{\pi_S}+\sum_{\mathbf{w}'\in\mathcal{I}}p^{\pi_S}(\mathbf{w}'|\mathbf{w},a_A)g_{\pi_A}^{\pi_S}\mathbf{(w}'),
	\end{aligned}
	\end{equation}
	\begin{equation}
	\begin{aligned}
	&Q_{\pi_A}^{\pi_S}(o_A,a_A)\triangleq{\mathop {\lim \sup }\limits_{T \to \infty } \frac{1}{T}{\mathbb{E}^{(\pi_S,\pi_A)} }\left[{\sum\limits_{t = 0}^{T - 1} \left(r(t)-\eta_{\pi_A}^{\pi_S}\right) }\left|o_A^{(0)}=o_A,a_A(0)=a_A\right. \right]}\\
	&=\sum_{\mathbf{w}\in\mathcal{I}}p_{\pi_A}^{\pi_S}(\mathbf{w}|o_A)Q_{\pi_A}^{\pi_S}(\mathbf{w},a_A),
	\end{aligned}
	\end{equation}
	where $p_{\pi_A}^{\pi_S}(\mathbf{w}'|\mathbf{w})$ can be obtained by (\ref{transitionp2}) and $p_{\pi_A}^{\pi_S}(\mathbf{w}|o_A)$ can be obtained by the Bayesian formula:
	\begin{equation}
	p_{\pi_A}^{\pi_S}(\mathbf{w}|o_A)=\frac{{\mu}_{\pi_A}^{\pi_S}(\mathbf{w})p(o_A|\mathbf{w})}{\sum_{\mathbf{w}\in\mathcal{I}}{\mu}_{\pi_A}^{\pi_S}(\mathbf{w})p(o_A|\mathbf{w})}.
	\end{equation}
\end{lemma}
\begin{proof}
	Please refer to Appendix \ref{prooflemma3}.
\end{proof}
With $Q_{\pi_A}^{\pi_S}(o_A,a_A)$ in hand, it is then easy to conduct the Policy Iteration (PI) Algorithm with Step Sizes \cite{li2011finding} to iteratively improve the deterministic memoryless decision-making policy $\pi_A$. The detailed steps are shown in Algorithm \ref{Algorithm 3}.

\begin{algorithm}
	\small
	\caption{The PI Algorithm with Step Size to Solve the POMDP Given  $\pi_S$}
	\label{Algorithm 3}
	\LinesNumbered
	\KwIn{The POMDP tuple $\mathscr{M}^{\pi_S}_{\mathrm{POMDP}}\triangleq\langle\mathcal{I},\mathcal{A}_A,\Omega_A,\mathcal{O}_A,\mathcal{T}^{\pi_S}_{\mathrm{POMDP}},\mathcal{R}\rangle$, $\epsilon$, $\pi_S$;}
	Initialization: randomly choose decision-making policy $\pi_A^{(1)}$, $\eta_{{\pi_A}^{(0)}}^{\pi_S}=0$, $\eta_{{\pi_A}^{(-1)}}^{\pi_S}=\infty$, $k=0$\;
	\While{$|\eta_{{\pi_A}^{(k)}}^{\pi_S}-\eta_{{\pi_A}^{(k-1)}}^{\pi_S}|\ge\epsilon$}{
	$k=k+1$\;
	Calculate the transition probability $p_{{\pi_A}^{(k)}}^{\pi_S}(\mathbf{w}'|\mathbf{w})$ by (\ref{transitionp2})\;
	Solve the stationary distribution $\boldsymbol{\mu}_{{\pi_A}^{(k)}}^{\pi_S}$ by the stationary equations (\ref{station})\;
	Calculate the expectation reward $r_{{\pi_A}^{(k)}}^{\pi_S}(\mathbf{w})$ by (\ref{stationreward}) and the long-term average reward $\eta_{{\pi_A}^{(k)}}^{\pi_S}$ by (\ref{averagereward})\;
	Calculate the relative reward $g_{{\pi_A}^{(k)}}^{\pi_S}(\mathbf{w})$ by (\ref{iterations})\;
	\For{$o_A\in\mathcal{O}_A$}{
		\For{$a_A\in\mathcal{A}_A$}{
				Calculate $Q_{{\pi_A}^{(k)}}^{\pi_S}(o_A,a_A)$ by Lemma 3\;}
	}
	\For{$o_A\in\mathcal{O}_A$}{
	$\pi_A(\cdot|o_A)=\mathop{\arg\max}\limits_{\pi_A(\cdot|o_A)}Q_{\pi_A^{(k)}}^{\pi_S}(o_A,a_A)$\;
	$\pi_A^{(k)}(\cdot|o_A)=\pi_A^{(k-1)}(\cdot|o_A)+\delta_k*(\pi_A^{(k-1)}(\cdot|o_A)-\pi_A(o_A))$\;}
	}
	\For{$o_A\in\mathcal{O}_A$}{
	$\pi_A^{*}(o_A)=\mathop{\arg\max}\limits_{a_A}\pi_A^{(k)}(a_A|o_A)$\;}
	$\theta^*(\pi_S,\pi_A^*)=\eta_{{\pi_A}^{*}}^{\pi_S}$\;
	\KwOut{$\pi_A^{*}(\pi_S)$, $\theta^*(\pi_S,\pi_A^*)$}
\end{algorithm}
Thus far, we have solved two problems: $i$) by capitalizing on Definition \ref{proposition1} and Algorithm \ref{Algorithm 1}, we have ascertained an optimal sampling strategy $\pi_S^*$ contingent upon the decision-making policy $\pi_A$; $ii$) by harnessing Definition \ref{proposition2} and Algorithm \ref{Algorithm 3}, we have determined an optimal actuation strategy $\pi_S^*$ predicated on the decision-making policy $\pi_A$. Consequently, we could iteratively employ Algorithm 1 and Algorithm \ref{Algorithm 3} in an alternating fashion, whereby Algorithm 1 yields the optimal sampling strategy $\pi_S^{*(k)}(\pi_A^{(k-1)})$, subsequently serving as an input for Algorithm 3 to derive the decision-making policy $\pi_A^{*(k)}(\pi_S^{*(k)})$. The procedure shall persist until the average reward $\theta^*(\pi_S^{*(k)},\pi_A^{*(k)})$ reaches convergence, indicating that the solution achieves a \emph{Nash equilibrium} between the sampler and the actuator.The intricacies of the procedure are delineated in Algorithm \ref{Algorithm 4}.
\begin{remark}
	Generally, the JESP algorithm should restart the algorithm by randomly choosing the initial decision-making policy $\pi_A^{*(1)}$ to ensure a good solution, as the initialization of decision-making policy $\pi_A^{*(1)}$ may often lead to poor local optima. We here investigate a heuristic initialization to find the solution quickly and reliably. Specifically, we assume that the decision-maker is fully observable and solve a MDP problem:
	\begin{definition}
		$\mathscr{M}_{\mathrm{MDP}}\triangleq\langle\mathcal{I},\mathcal{A}_A,\mathcal{T}_{\mathrm{MDP}},\mathcal{R}\rangle$, where the elements are given as follows:
		\begin{itemize}
			\item $\mathcal{I}$: the set of $(X_t,\Phi_t)\in\mathcal{S}\times\mathcal{V}$.
			\item $\mathcal{A}_A=\left\{a_0,a_{M-1}\right\}$: the decision-making set.
			\item $\mathcal{T}$: the transition function, given as
			\begin{equation}
			\begin{aligned}
			&\mathcal{T}((X_t,\Phi_t),a_A,(X_{t+1},\Phi_{t+1}))\\
			&=p(X_{t+1}|X_t,a_A,\Phi_{t})\cdot p(\Phi_{t+1}|\Phi_{t}),
			\end{aligned}
			\end{equation}
			where $p(X_{t+1}|X_t,a_A,\Phi_{t})$ and $p(\Phi_{t+1}|\Phi_{t})$ could be obtained by (\ref{Source}) and (\ref{context}), respectively.
			\item $\mathcal{R}$: the same as the pre-defined Dec-POMDP tuple.
		\end{itemize}
	\end{definition}
\end{remark}
Through solving the above MDP problem, we could explicitly obtain the Q function $Q(X_t,\Phi_t,a_A)$, define $Q(X_t,a_A)$ as ${\mathbb{E}}_{\Phi_{t}}[Q(X_t,\Phi_t,a_A)]$, the initial decision-making policy $\pi_A^{*(1)}$ is given as
\begin{equation}\label{piaini}
	\pi_A^{*(1)}(\hat{X}_t)={\arg\min}_{a_A}Q(\hat{X}_t,a_A).
\end{equation} 
\begin{algorithm}\small
	\caption{The Improved JESP Algorithm}
	\label{Algorithm 4}
	\LinesNumbered
	\KwIn{The Dec-POMDP tuple $\mathscr{M}_{DEC-POMDP}\triangleq\left\langle n, \mathcal{I}, \mathcal{A}, \mathcal{T}, \Omega, \mathcal{O}, \mathcal{R}  \right\rangle$, $\epsilon$;}
	Initialization: $\theta^{*}_0=0$, $\theta^{*}_{-1}=\infty$, $k=0$\;
	Initialize $\pi_A^{*(1)}$ by calculating (\ref{piaini})\;
	\While {$||\theta^{*}_k-\theta^{*}_{k-1}||\ge \epsilon$}
	{$k=k+1$\;
	Formulate the MDP problem $\mathscr{M}^{\pi_A^{*(k)}}_{\mathrm{MDP}}\triangleq\langle\mathcal{I},\mathcal{A}_S,\mathcal{T}_{\mathrm{MDP}}^{\pi_A^{*(k)}},\mathcal{R}\rangle$ as given in Definition \ref{proposition1}\;
	Run Algorithm 1 to obtain $\pi_S^{*(k)}$\;
	Formulate the POMDP problem $\mathscr{M}^{\pi_S^{*(k)}}_{\mathrm{POMDP}}\triangleq\langle\mathcal{I},\mathcal{A}_A,\Omega_A,\mathcal{O}_A,\mathcal{T}^{\pi_S^{*(k)}}_{\mathrm{POMDP}},\mathcal{R}\rangle$ as given in Definition \ref{proposition2}\;
	Run Algorithm 3 to obtain $\theta^*(\pi_S^{*(k)},\pi_A^{*(k)})$ and $\pi_A^{*(k)}$\;
	$\theta^{*}_k=\theta^*(\pi_S^{*(k)},\pi_A^{*(k)})$\;		
	}
	The joint sub-optimal policy is:
	$\pi_S^*=\pi_S^{*(k)},\pi_S^*=\pi_A^{*(k)}$\;
	The sub-optimal average reward is: $\theta^*=\theta^{*}_k$\;
	\KwOut{$\pi_S^{*}$, $\pi_A^*$, $\theta^*$}
\end{algorithm}
\section{Simulation Results}\label{sectionV}

Tradition metrics such as Age of Information have been developed under the assumption that a fresher packet or more accurate packet, capable of aiding in source reconstruction, holds a higher value for the receiver, thus promoting goal-oriented decision-making. Nevertheless, the manner in which a packet update impacts the system's \emph{utility} via decision-making remains an unexplored domain. Through the simulations, we endeavor to elucidate the following observations of interest:

\noindent $\bullet$ \emph{GoT-optimal vs. State-of-the-art}. In contrast with the state-of-the-art sampling policies, the proposed goal-oriented \emph{sampler \& decision-maker} co-design is capable of concurrently maximizing goal attainment and conserving communication resources, accomplishing a closed-loop \emph{utility} optimization via sparse sampling. (See Fig. \ref{RatevsCostComparision} and \ref{Costcomparison})

\noindent $\bullet$ \emph{Separate Design vs. Co-Design}. Compared to the two-stage sampling-decision-making separate framework, the co-design of sampling and decision-making not only achieves superior goal achievement but also alleviates resource expenditure engendered by communication and actuation implementation. (See Fig. \ref{RatevsCostComparision} and \ref{Costcomparison})

\noindent $\bullet$ \emph{Optimal Brute-Force-Search vs. Sub-optimal JESP}. {Under different successful transmission probability $p_S$ and sampling cost $C_S$, the sub-optimal yet computation-efficient JESP algorithm will converge to near-optimal solutions. }(See Fig. \ref{Ovssub})

\noindent $\bullet$ \emph{Trade-off: Transmission vs. Actuation}. There is a trade-off between transmission and actuation in terms of resource expenditure: under reliable channel conditions, it is apt to increase communication overhead to ensure effective decision-making; conversely, under poor channel conditions, it is advisable to curtail communication expenses and augment actuation resources to attain maximal system \emph{utility}. (See Fig. \ref{tradeoff})

\subsection{Comparing Benchmarks}\label{sectionVA}
Fig. \ref{RatevsCostComparision} illustrates the simulation results, which characterizes the \emph{utility} by the average cost composed by status inherent cost $C_1(X_t,\Phi_t)$, actuation gain cost $C_2(\pi_A(\hat{X}_t))$, actuation inherent cost $C_3(\pi_A(\hat{X}_t))$, and sampling cost $C_S$. For the simulation setup, we set $\mathcal{A}_A=\left\{a_0,\cdots,a_{10}\right\}$, $\mathcal{S}=\left\{s_0,s_1,s_2\right\}$, $\mathcal{V}=\left\{v_0,v_1\right\}$ and the corresponding costs are given as:
\begin{equation}
\mathbf{C}_1^{|\mathcal{V}|\times|\mathcal{S}|}=\begin{pmatrix}
0 &20 &50\\
0 &10 &20
\end{pmatrix},
\end{equation}
$C_2(\pi_A(\hat{X}_t))$ and $C_3(\pi_A(\hat{X}_t))$ are both linear to the actuation with $C_2(\pi_A(\hat{X}_t)=C_g\cdot\pi_A(\hat{X}_t)$ and $C_3(\pi_A(\hat{X}_t))=C_I\cdot\pi_A(\hat{X}_t)$. The following comparing benchmarks of interest are considered:

\noindent $\bullet$ {\textbf{Uniform.}} Sampling is triggered periodically in this policy. In this case, $a_S(t)=\mathbbm{1}_{\left\{t=K*\Delta\right\}}$, where $K=0,1,2,\cdots$ and $\Delta\in\mathbb{N}^+$. For each $\Delta$, we could calculate the sampling rate as $1/\Delta$ and explicitly obtain the long-term average cost through Markov chain simulations under pre-defined decision-making policy $\pi_A=[a_0,a_3,a_7]$. The setup of $\Delta$ represents a trade-off between \emph{utility} and sampling frequency, as depicted in Fig. \ref{RatevsCostComparision}: If $C_S$ is minimal, sampling and transmission will contribute positively to the \emph{utility}; If the sampling action is expansive, sampling may not yield adequate \emph{utility}; If a single sampling consumes moderate resource, the \emph{utility} will exhibit a U-shaped pattern in terms of the sampling rate.

\noindent $\bullet$ {\textbf{Age-aware.}} Sampling is executed when the AoI attains a predetermined threshold, a principle that has been established as a threshold-based result for AoI-optimal sampling \cite{sun2019sampling}. In this case, $a_S(t)=\mathbbm{1}_{\left\{\mathrm{AoI(t)}>\delta\right\}}$, where the AoI-optimal threshold $\delta$ can be ascertained using the Bisection method delineated in Algorithm 1 of \cite{sun2019sampling}. In this context, rather than determining a fixed threshold that minimized AoI, we dynamically shift the threshold to explore the balance between sampling and \emph{utility}. As evidenced in Fig. \ref{RatevsCostComparision}, the \emph{utility} derived from this sampling policy consistently surpasses that of the uniform sampling policy.

\noindent $\bullet$ {\textbf{Change-aware.}} Sampling is triggered whenever the source status changes. In such a case, $a_S(t)=\mathbbm{1}_{\left\{X_t\ne X_{t-1}\right\}}$. The performance of this policy is dependent on the dynamics of the system, \emph{i.e.}, if the semantics of the sources transfers frequently, then the sampling rate will be higher. The \emph{utility} of this policy may be arbitrarily detrimental owing to its semantics-unaware nature. In our considered model, the Change-aware policy will turn out $89.5\%$ to have the unsynchronized status $X_t=s_0$ while $\hat{X}_t=s_3$. In this case, the actuator will implement the actuation $a_A=a_7$ according to the estimate $\hat{X}_t=s_3$, which will in turn make the status $X_t$ converges to $s_0$.  

\noindent $\bullet$ {\textbf{Optimal MSE.}} This is a type of E2E goal-oriented sampling policy if the goal is determined as achieving real-time reconstruction. Nevertheless, this policy disregards the semantics conveyed by the packet and the ensuing actuation updating precipitated by semantics updates. The problem could be formulated as a standard MDP formulation and solved out through RVI Algorithm. The sampling rate and average cost are obtained given the MSE-optimal sampling policy and the pre-defined decision-making policy $\pi_A=[a_0,a_3,a_7]$. 

\noindent $\bullet$ {\textbf{Optimal AoII (also optimal AoCI).}} From \cite{AoII}, it has been proven that the AoII-optimal sampling policy turns out to be $a_S(t)=\mathbbm{1}_{\left\{X_t\ne \hat{X}_t \right\}}$. From \cite{AoCI}, the AoCI-optimal sampling policy is $a_S(t)=\mathbbm{1}_{\left\{X_t\ne {X}_{t-\mathrm{AoI}(t)} \right\}}$. Note that $\hat{X}_t={X}_{t-\mathrm{AoI}(t)}$, these two sampling policies are equivalent. The sampling rate and average cost are obtained given this sampling policy and the greedy-based decision-making policy. $\pi_A=[a_0,a_3,a_7]$.
\subsection{Separate Design \emph{vs.} Co-Design}\label{sectionVB}
Conventionally, the sampling and actuation policies are designed in a two-stage manner: they first emphasize open-loop performance metrics such as average mean squared error (MSE) or average Age of Information (AoI), we then focus on the decision-making policy $\pi_A$ design. Specifically, we consider that $\pi_A$ is predetermined using a greedy methodology:
\begin{equation}\label{greedy}
\pi_A(\hat{X}_t)=\mathop{\arg\min}\limits_{a_A\in\mathcal{S}_A}\mathop{\mathbb{E}}\limits_{\Phi_{t}}\left\{\left[C_1(\hat{X}_t,\Phi_t) - C_2(\pi_A(\hat{X}_t))\right]^+ C_3(\pi_A(\hat{X}_t))\right\}.
\end{equation}
This greedy approach entails selecting the actuation that minimizes cost in the current step, given that the estimate $\hat{X}_t$ is perfect. By calculating (\ref{greedy}), we obtain $\pi_A = [a_0,a_3,a_7]$.

\begin{figure}
		\centering
	\begin{minipage}[t]{0.48\textwidth}
			\centering	\includegraphics[angle=0,width=1\textwidth]{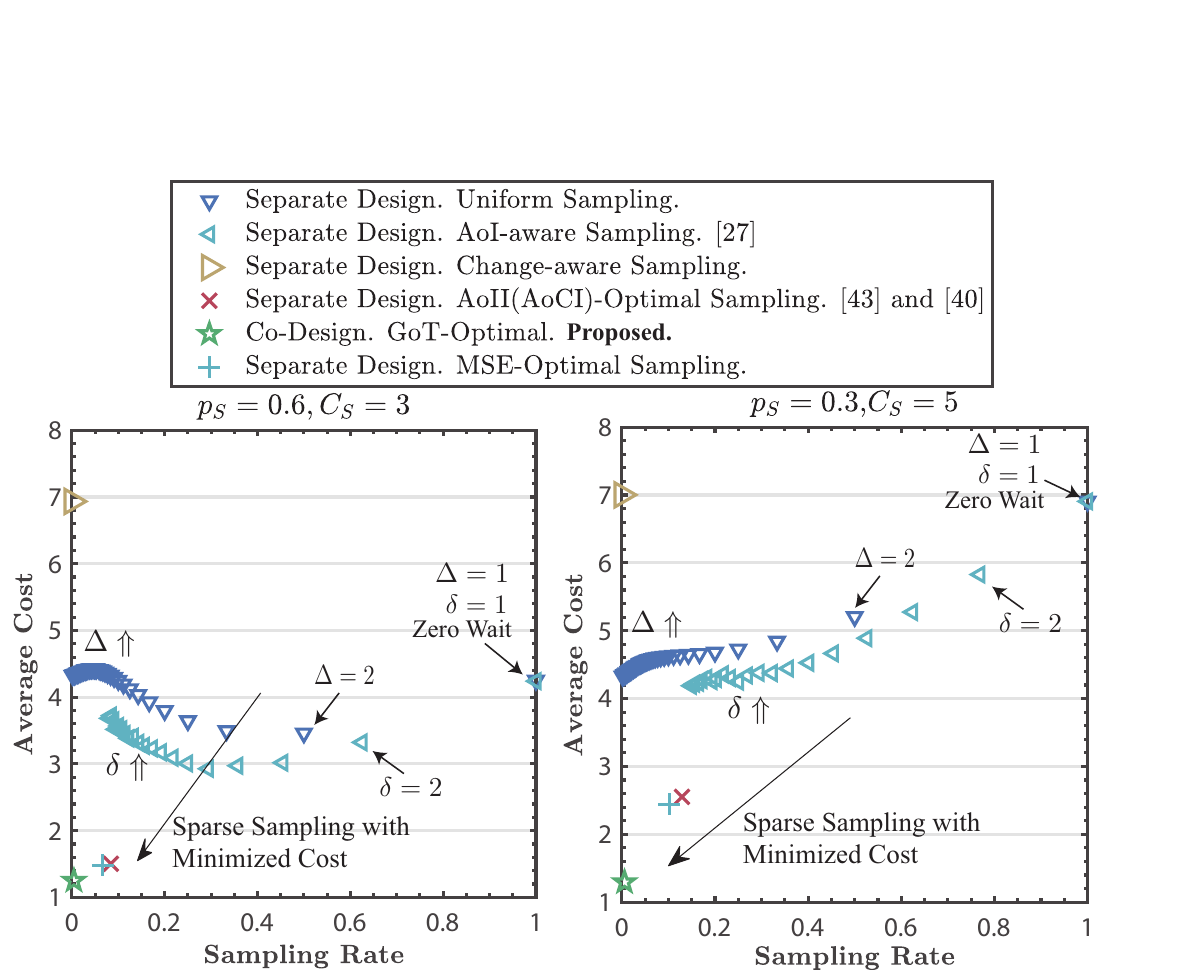}
		\caption{Average Cost vs. Sampling Rate under different policies and parameters setup. Here we set $C_g=8$ and $C_I=1$.}\label{RatevsCostComparision}
	\end{minipage}
	\begin{minipage}[t]{0.48\textwidth}
			\centering
		\includegraphics[angle=0,width=1\textwidth]{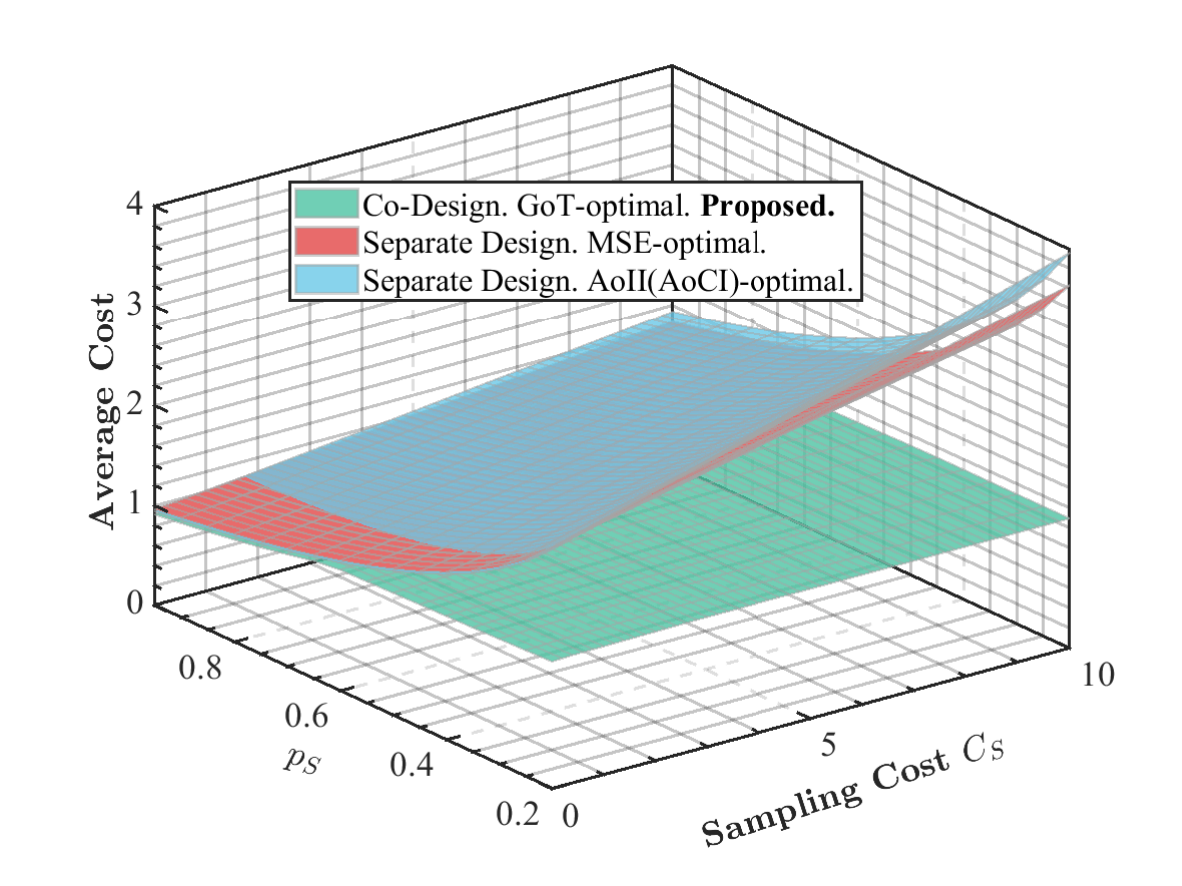}
		\caption{Comparisons among GoT-optimal, AoII-optimal, and MSE-optimal policies. Here we set $C_g=8$ and $C_I=1$.}\label{Costcomparison}
	\end{minipage}
	\begin{minipage}[t]{0.48\textwidth}
	\centering
\includegraphics[angle=0,width=1\textwidth]{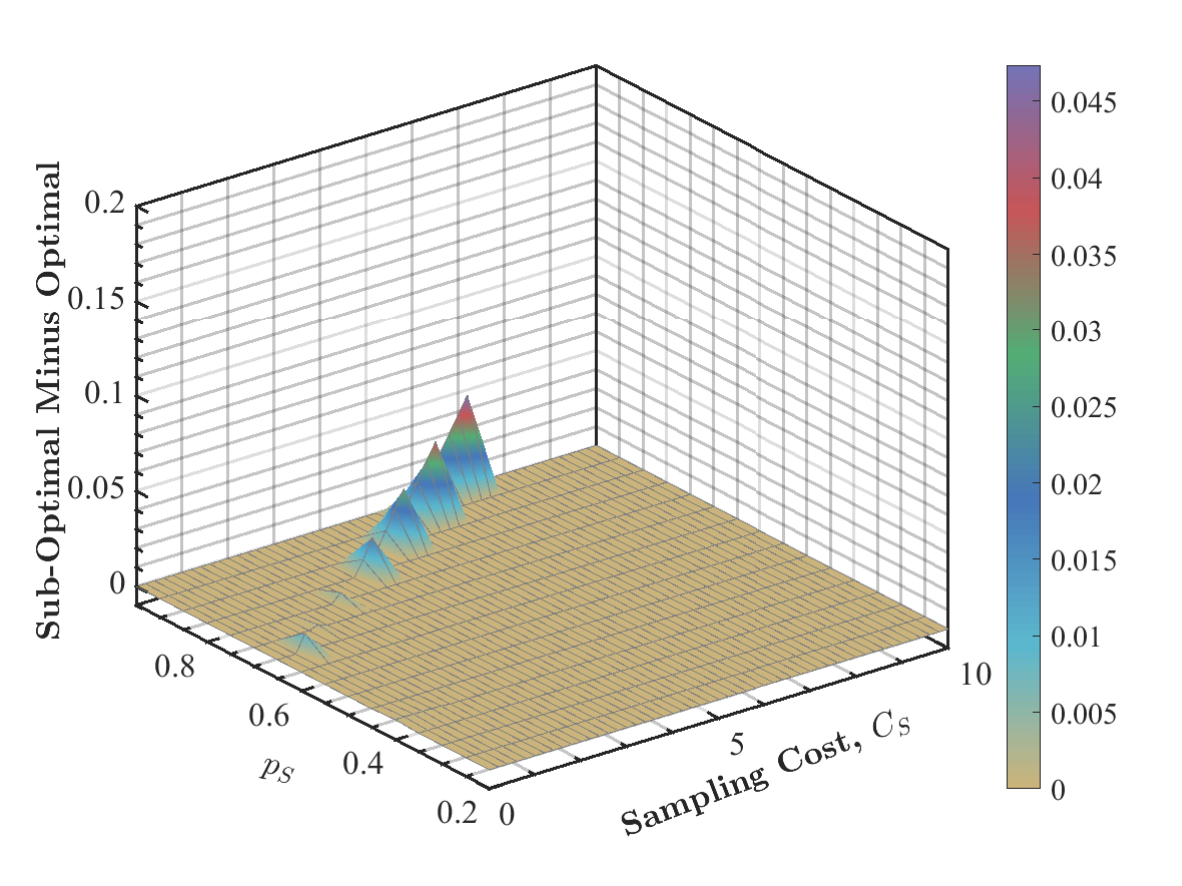}
\caption{Distance between optimal solutions through RVI-Brute-Force-Search algorithm and Sub-optimal solutions through JESP algorithm.}\label{Ovssub}
\end{minipage}
\begin{minipage}[t]{0.48\textwidth}
		\centering
	\includegraphics[angle=0,width=1\textwidth]{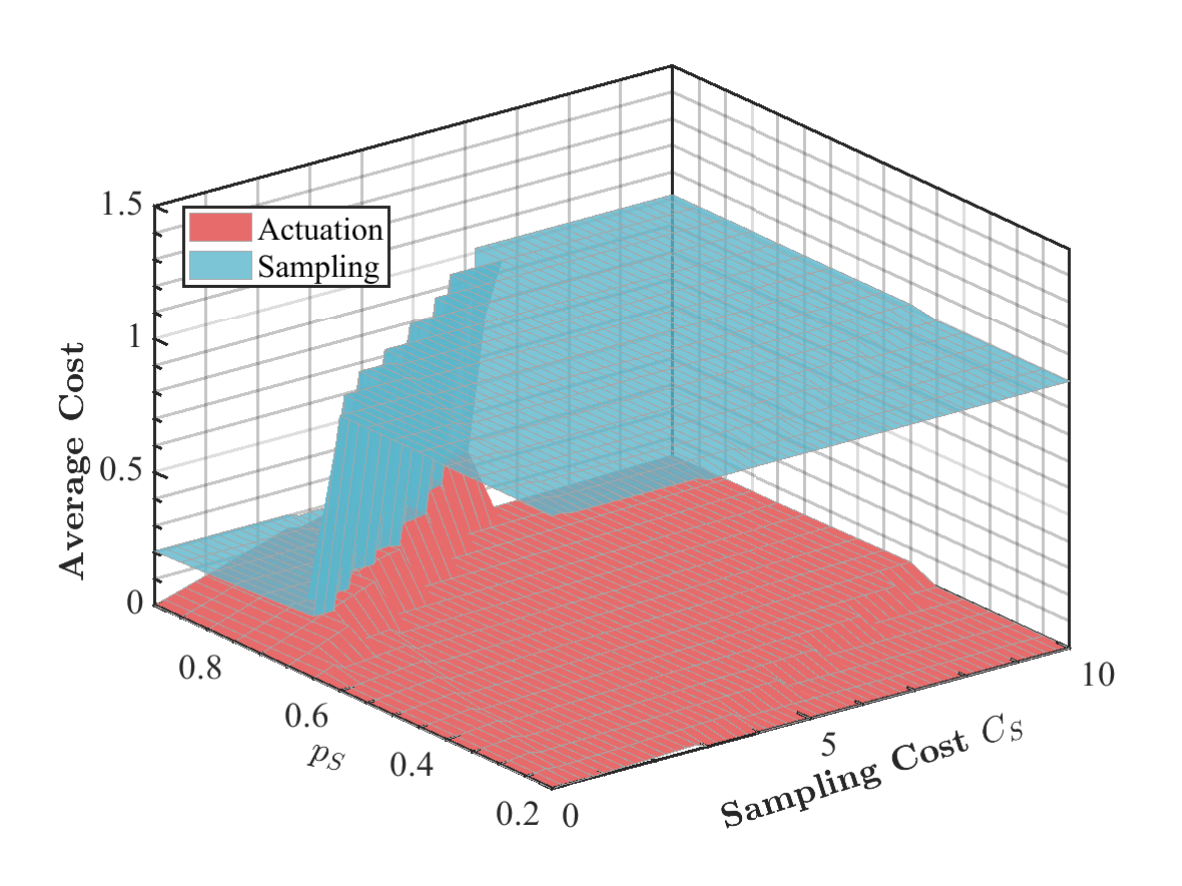}
	\caption{Trade-off between long-term average transmission cost and long-term average actuation cost.}\label{tradeoff}
\end{minipage}
\end{figure}

However, we notice that sampling and actuation are closely intertwined, highlighting the potential for further co-design. In this paper, we have proposed the RVI-Brute-Force-Search and the Improved JESP algorithms for such optimal co-design. As shown in Fig. \ref{RatevsCostComparision}, the \emph{sampler \& decision-maker} co-design achieves the optimal utility through sparse sampling. Specifically, only semantically important information is sampled and transmitted, while non-essential data is excluded. This goal-oriented, semantic-aware, and sparse sampling design represents a significant advancement in sampling policy design. By incorporating a best-matching decision-making policy, the sparse sampling achieves superior performance compared to existing methods.

{Fig. \ref{Costcomparison} presents a comparative analysis between the AoII (or AoCI)-optimal sampling policy, the MSE-optimal sampling policy, and our proposed GoT-optimal \emph{sampler \& decision-maker} co-design. It is verified that under different $p_S$ and $C_S$, the proposed GoT-optimal \emph{sampler \& decision-maker} co-design achieves the optimal goal-oriented \emph{utility}. Importantly, under the condition of an extremely unreliable channel $p_S=0.2$ and high sampling cost $C_S=10$, the proposed co-design facilitates a significant reduction in long-term average cost, exceeding 60\%. This underscores the superiority of the GoT-optimal \emph{sampler \& decision-maker} co-design.}

\subsection{Optimal vs. Sub-Optimal}
Fig. \ref{Ovssub} presents a comparative visualization between optimal and sub-optimal solutions over a wide range of $C_S$ and $p_S$ values. The negligible zero-approaching value in Fig. \ref{Ovssub} implies a trivial deviation between the optimal and sub-optimal solutions, suggesting the latter's potential for convergence towards near-optimal outcomes. The minimal variance testifies to the sub-optimal algorithm's consistent ability to approximate solutions with high proximity to the optimal. This critical observation underscores the practical advantages of employing sub-optimal improved JESP Algorithm, especially in scenarios with extensive $\mathcal{A}_A$ and $\mathcal{O}_A$.
\subsection{Trade-off: Transmission vs. Actuation}\label{sectionVC}
Fig. \ref{tradeoff} exemplifies the resource allocation trade-off between transmission and actuation when the long-term average cost is minimized. When the probability of $p_S$ remains low (signifying an unreliable channel) or $C_S$ is high (indicating expensive sampling), it becomes prudent to decrease sampling and transmission, while concurrently augmenting actuation resources for optimal system \emph{utility}. In contrast, when the channel is reliable, sampling and transmission resource can be harmonized with actuation resources to achieve the goal better. This indicates that through the investigation of the optimal co-design of the \emph{sampler \& decision-maker} paradigm, a trade-off between transmission and actuation resources can be achieved.

\section{Conclusion}\label{sectionVI}
In this paper, we have investigated the GoT metric to directly describe the goal-oriented system decision-making \emph{utility}. Employing the proposed GoT, we have formulated an infinite horizon Dec-POMDP problem to accomplish the co-design of sampling and actuating. To address this problem, we have developed two algorithms: the computationally intensive RVI-Brute-Force-Search, which is proven to be optimal, and the more efficient, albeit suboptimal algorithm, named JESP Algorithm. Comparative analyses have substantiated that the proposed GoT-optimal \emph{sampler \& decision-maker} pair can achieve sparse sampling and meanwhile maximize the \emph{utility}, signifying the initial realization of a sparse, goal-oriented, and semantics-aware sampler design.

\bibliographystyle{IEEEtran}
\bibliography{reference}

\begin{appendices}
	\section{The Proof of Lemma \ref{as11}}\label{prooflemma1}
	By taking into account the \textit{conditional independence} among $X_{t+1}$, $\Phi_{t+1}$, and $X_{t+1}$, given $(X_{t},\Phi_{t},X_{t})$ and $\mathbf{a}(t)$, we we can express the following:
	
\begin{equation}\label{wwt}
\small
\begin{aligned}
&\Pr\left\{\mathbf{W}_{t+1}=(s_u,x,v_r)\left|\mathbf{W}_t=(s_i,s_j,v_k),\mathbf{a}(t)=(1,a_m)\right.\right\}\\
&=\Pr\left\{X_{t+1}=s_u\left|\mathbf{W}_t=(s_i,s_j,v_k),\mathbf{a}(t)=(1,a_m)\right.\right\}\times\\
&\Pr\left\{\hat{X}_{t+1}=x\left|\mathbf{W}_t=(s_i,s_j,v_k),\mathbf{a}(t)=(1,a_m)\right.\right\}\times\\
&\Pr\left\{\Phi_{t+1}=v_r\left|\mathbf{W}_t=(s_i,s_j,v_k),\mathbf{a}(t)=(1,a_m)\right.\right\},
\end{aligned}
\end{equation}
wherein the first, second, and third terms can be derived through \emph{conditional independence}, resulting in simplified expressions of (\ref{Source}), (\ref{as1}), and (\ref{context}), respectively:
\begin{equation}\label{1}
\small
\begin{aligned}
	\Pr\left\{X_{t+1}=s_u\left|\mathbf{W}_t=(s_i,s_j,v_k),\mathbf{a}(t)=(1,a_m)\right.\right\}=p_{i,u}^{(k,m)},
\end{aligned}
\end{equation}
\begin{equation}\label{2}
\small
\begin{aligned}
	\Pr\left\{\hat{X}_{t+1}=x\left|\mathbf{W}_t=(s_i,s_j,v_k),\mathbf{a}(t)=(1,a_m)\right.\right\}=p_S\cdot\mathbbm{1}_{\left\{x=s_i\right\}}+(1-p_S)\cdot\mathbbm{1}_{\left\{x=s_j\right\}},
\end{aligned}
\end{equation}
\begin{equation}\label{3}
\small
\begin{aligned}
\Pr\left\{\Phi_{t+1}=v_r\left|\mathbf{W}_t=(s_i,s_j,v_k),\mathbf{a}(t)=(1,a_m)\right.\right\}=p_{k,r},
\end{aligned}
\end{equation}
Substituting (\ref{1}), (\ref{2}), and (\ref{3}) into (\ref{wwt}) yields the (\ref{as11}) in Lemma \ref{as11}.
In the case where $a_S(t)=0$, we can obtain a similar expression by replacing $\mathbf{a}(t)=(1,m)$ with $\mathbf{a}(t)=(0,m)$. Substituting (\ref{Source}), (\ref{as0}), and (\ref{context}) into this new expression results in the proof of (\ref{as00}) in Lemma \ref{as11}.

	\section{Proof of Lemma 3}\label{prooflemma3}
	$Q_{\pi_A}^{\pi_S}(\mathbf{w},a_A)$ could be simplified as follows:
	\begin{equation}
	\small
	\begin{aligned}
	&{\mathop {\lim \sup }\limits_{T \to \infty } \frac{1}{T}{\mathbb{E}^{(\pi_S,\pi_A)} }\left[{\sum\limits_{t = 0}^{T - 1} \left(r(t)-\eta_{\pi_A}^{\pi_S}\right) }\left|\mathbf{W}_0=\mathbf{w},a_A(0)=a_A\right. \right]}\\
	&=\mathcal{R}^{\pi_S}(\mathbf{w},a_A)-\eta_{\pi_A}^{\pi_S}+\sum_{\mathbf{w}'\in\mathcal{I}}p^{\pi_S}(\mathbf{w}'|\mathbf{w},a_A)\cdot\underbrace{{\mathop {\lim \sup }\limits_{T \to \infty } \frac{1}{T-1}{\mathbb{E}^{(\pi_S,\pi_A)} }\left[{\sum\limits_{t = 1}^{T - 1} \left(r(t)-\eta_{\pi_A}^{\pi_S}\right) }\left|\mathbf{W}_1=\mathbf{w}'\right. \right]}}_{g_{\pi_A}^{\pi_S}\mathbf{(w}')}\\
	&=\mathcal{R}^{\pi_S}(\mathbf{w},a_A)-\eta_{\pi_A}^{\pi_S}+\sum_{\mathbf{w}'\in\mathcal{I}}p^{\pi_S}(\mathbf{w}'|\mathbf{w},a_A)g_{\pi_A}^{\pi_S}\mathbf{(w}'),\\
	\end{aligned}
	\end{equation}
	$Q_{\pi_A}^{\pi_S}(o_A,a_A)$ could be solved as:
	\begin{equation}
	\small
	\begin{aligned}
	&{\mathop {\lim \sup }\limits_{T \to \infty } \frac{1}{T}{\mathbb{E}^{(\pi_S,\pi_A)} }\left[{\sum\limits_{t = 0}^{T - 1} \left(r(t)-\eta_{\pi_A}^{\pi_S}\right) }\left|o_A^{(0)}=o_A,a_A(0)=a_A\right. \right]}\\
	&=\sum_{\mathbf{w}\in\mathcal{I}}p_{\pi_A}^{\pi_S}(\mathbf{w}|o_A,a_A)\cdot\underbrace{{\mathop {\lim \sup }\limits_{T \to \infty } \frac{1}{T-1}{\mathbb{E}^{(\pi_S,\pi_A)} }\left[{\sum\limits_{t = 0}^{T - 1} \left(r(t)-\eta_{\pi_A}^{\pi_S}\right) }\left|\mathbf{W}_0=\mathbf{w},a_A(0)=a_A\right. \right]}}_{Q_{\pi_A}^{\pi_S}(\mathbf{w},a_A)}\\
	&=\sum_{\mathbf{w}\in\mathcal{I}}p_{\pi_A}^{\pi_S}(\mathbf{w}|o_A,a_A)Q_{\pi_A}^{\pi_S}(\mathbf{w},a_A).
	\end{aligned}
	\end{equation}
	where $p_{\pi_A}^{\pi_S}(\mathbf{w}|o_A,a_A)$ is the posterior conditional probability. Note that the state $\mathbf{w}$ is independent of $a_A$ when $o_A$ is known, we have that
	\begin{equation}
	\small
	\begin{aligned}
		&p_{\pi_A}^{\pi_S}(\mathbf{w}|o_A,a_A)=p_{\pi_A}^{\pi_S}(\mathbf{w}|o_A)=\frac{p_{\pi_A}^{\pi_S}(\mathbf{w},o_A)}{p_{\pi_A}^{\pi_S}(o_A)}=\frac{\boldsymbol{\mu}_{\pi_A}^{\pi_S}(\mathbf{w})p(o_A|\mathbf{w})}{\sum_{\mathbf{w}\in\mathcal{I}}\boldsymbol{\mu}_{\pi_A}^{\pi_S}(\mathbf{w})p(o_A|\mathbf{w})}.
	\end{aligned}
	\end{equation}

\end{appendices}

\end{document}